\documentclass[10pt, a4paper]{IEEEtran}
\usepackage{epsfig,endnotes}

\usepackage{xcolor}
\usepackage{amsmath, amsthm, amssymb}
\usepackage{wrapfig}
\usepackage{cite}

\newcommand{\pbgnew}{\textcolor{black}}
\newcommand{\pbgnewer}{\textcolor{black}}
\newcommand{\paragraphb}[1]{\medskip{} \noindent \textbf{#1}}
\newcommand{\paragraphbi}[1]{\medskip{} \noindent \textit{\textbf{#1}}}
\newcommand{\cut}[1]{}

\usepackage{numprint}
\usepackage{url}
\newcolumntype{L}[1]{>{\raggedright\let\newline\\\arraybackslash\hspace{0pt}}m{#1}}
\newcolumntype{C}[1]{>{\centering\let\newline\\\arraybackslash\hspace{0pt}}m{#1}}
\newcolumntype{R}[1]{>{\raggedleft\let\newline\\\arraybackslash\hspace{0pt}}m{#1}}

\newtheorem{theorem}{Theorem}
\newtheorem{claim}[theorem]{Claim}

\usepackage{xspace}
\newcommand{\ie}{{\it i.e.,}\xspace}

\newcommand{\ajsnew}{\textcolor{black}}
\newcommand{\SC}{\textcolor{black}}
\newcommand{\fixme}[1]{\textcolor{red}{[\textbf{#1}]}}

\usepackage{subfigure}
\usepackage{enumitem}
\usepackage{listings}
\usepackage{algorithm} 
\usepackage{algpseudocode}
\usepackage{listings}
\usepackage{setspace}

\begin{document}
%



\title{Measuring and Understanding Throughput\\of Network Topologies }



\author{
\IEEEauthorblockN{Sangeetha Abdu Jyothi\IEEEauthorrefmark{1}, Ankit Singla\IEEEauthorrefmark{2}, P. Brighten Godfrey\IEEEauthorrefmark{1}, Alexandra Kolla\IEEEauthorrefmark{1}} \\
    \IEEEauthorblockA{\IEEEauthorrefmark{1}University of Illinois at Urbana--Champaign    }
    \IEEEauthorblockA{\IEEEauthorrefmark{2}ETH Zurich}

}


\maketitle

\begin{abstract}

High throughput is of particular interest in data center and HPC networks. Although myriad network topologies have been proposed, a broad head-to-head comparison across topologies and across traffic patterns is absent, and the right way to compare worst-case throughput performance is a subtle problem.

In this paper, we develop a framework to benchmark the throughput of network topologies, using a two-pronged approach. First, we study performance on a variety of synthetic and experimentally-measured traffic matrices (TMs).  Second, we show how to measure worst-case throughput by generating a near-worst-case TM for any given topology.  We apply the framework to study the performance of these TMs in a wide range of network topologies, revealing insights into the performance of topologies with scaling, robustness of performance across TMs, and the effect of scattered workload placement. Our evaluation code is freely available.
\end{abstract}

\section{Introduction}
\label{sec:intro}

Throughput is a fundamental property of communication networks: at what rate can data be carried across the network between desired end-points? Particularly for data centers and high performance computing, an increase in throughput demand among compute elements has reinvigorated research on network topology, and a large number of network topologies have been proposed in the past few years to achieve high capacity at low cost~\cite{fattree-new, vl2, BCube, portland, dcell, cthrough, helios, proteus, jellyfish, legup, rewire, Xpander}.

However, there is little order to this large and ever-growing set of network topologies. We lack a broad comparison of topologies, and there is no open, public framework available for testing and comparing topology designs.  The absence of a well-specified benchmark complicates research on network design, making it difficult to evaluate a new design against the numerous past proposals, and difficult for industry to know which threads of research are most promising to adopt.

Our goal is to build a framework for accurate and consistent measurement of the throughput of network topologies, and use this framework to benchmark proposed data center and HPC topologies.

To accomplish this, we need metrics for comparison of throughput, and this turns out to be a subtle problem. Throughput can be measured by testing particular workloads, or traffic matrices (TMs), but the immediate question is what TMs to test.  One approach is to test a variety of common TMs, which can provide insight into the effect of topological structure for particular use cases reflected by those specific TMs.  However, we argue it is useful to go beyond this. In HPC and data center networks, TMs may vary widely depending on the use-case and across time as applications spin up and down or migrate~\cite{vl2, hpcContention, hpcQos, netNoise, hpcInterferenceAvoidance}. Although some applications map well onto certain topologies and known worst-case traffic patterns~\cite{tornado, dragonflyStencil} can be avoided in such cases, a mix of multiple applications could still produce an unintended difficult TM. In fact, \cite{hpcContention} observes that network contention among applications sharing an HPC system will worsen in the near future (even after accounting for the expected increase in network speeds) as HPC applications grow in size and variety. Hence, it is useful to understand the \textit{worst-case} performance of a topology, for any TM.  However, currently, there does not exist a systematic way to evaluate the worst-case throughput achievable in a network and to identify the traffic pattern responsible for it.

Our key contributions, then, are to (1) develop a heuristic to measure worst-case throughput, and (2) provide an expansive benchmarking of a variety of topologies using a variety of TMs -- TMs generated from real-world measurements, synthetic TMs, and finally our new near-worst-case TMs.  We discuss each of these in more detail.

(1a) \textbf{We evaluate whether cut-based metrics}, such as bisection bandwidth and sparsest cut, solve the problem of estimating worst-case throughput. A number of studies (e.g.~\cite{rewire,ANCS, webb2011topology, stephens2012past, hyperx}) employ cut metrics. It has been noted~\cite{LFTI} that bisection bandwidth does not always predict average-case throughput (in a limited setting; \S\ref{sec:related}).  But do cuts measure \emph{worst-case} throughput?  We show that it does not, by proving the existence of two families of networks, $A$ and $B$, where $A$ has a higher cut-metric even though $B$ supports \emph{asymptotically} higher worst-case throughput.  Further, we show that the mismatch between cuts and worst-case throughput exists even for highly-structured networks of small size -- a 5-ary 3-stage butterfly with only 25 nodes -- where the sparsest-cut found through brute force computation is strictly greater than the worst-case throughput.

(1b) Since cut metrics don't achieve our goal, \textbf{we develop a heuristic to measure worst-case throughput}. We propose an efficient algorithm to generate a near-worst-case TM for any given topology. We show empirically that these near-worst-case TMs approach a theoretical lower bound of throughput in practice. Note that Kodialam et. al.~\cite{kodialam} previously developed a TM for a somewhat different reason\footnote{Upper-bounding throughput performance of a routing scheme using an LP formulation. In addition, \cite{kodialam} did not use their TM to benchmark topologies.} which could be re-purposed as a near-worst-case TM.  \pbgnewer{Compared with~\cite{kodialam}, our methodology finds TMs that are just as close to the worst case, can be computed approximately $6\times$ faster, and scales to networks $8\times$ larger with the same memory limit.}

(2) We perform a \textbf{head-to-head benchmark of a wide range of topologies} across a variety of workloads.  Specifically, we evaluate a set of $10$ topologies proposed for data centers and high performance computing, with our near-worst-case TMs, a selection of synthetic TMs, and real TMs measured from two operational Facebook data centers.  Key findings of this study include:

\begin{itemize}[leftmargin=*]
	\item \textbf{Which topologies perform well?} We find the answer depends on scale.  At small scale, several topologies, particularly DCell, provide excellent performance across nearly all TMs.  However, at $\geq 1000$s of servers, proposals based on expander graphs (Jellyfish, Long Hop, Slim Fly) provide the best performance, with the advantage increasing with scale (an observation confirmed by a recent data center proposal~\cite{Xpander}).
	
	These results add new information compared with past work.  We find the three expanders have nearly identical performance for uniform traffic (e.g. all-to-all); this differs from conclusions of~\cite{ANCS,slim-fly} which analyzed bisection bandwidth rather than throughput.  We find that expanders significantly outperform fat trees; this differs from the results of~\cite{HPC13} which concluded that fat trees perform as well or better.  We replicate the result of~\cite{HPC13} and show the discrepancy results from an inaccurate calculation of throughput and unequal topology sizes. 
	
	\item \textbf{What topologies have robustly high performance even with worst-case traffic?} The effect of near-worst-case traffic differs. We find that certain topologies (Dragonfly, BCube) are less robust, dropping more substantially in performance. Jellyfish and Long Hop are most robust.
	
	While fat trees are otherwise generally robust, we show that \pbgnewer{under a TM with a small fraction of large flows in the network}, fat trees perform poorly. We find this is due to lack of path diversity near the endpoints.
	
	\item \textbf{How do real-world workloads perform?} \pbgnewer{We experiment with real world TMs from two different clusters based on data from Facebook~\cite{FBTM}. Surprisingly, in one very nonuniform TM} we find that randomizing rack-level workload placement consistently \emph{improves} performance for all networks we tested other than fat trees and the expanders. The fact that expanders and fat trees perform well even without this workload rearrangement is consistent with our finding that they are robust to worst-case performance.  But also, for the other topologies, it indicates an intriguing opportunity for improving performance through better workload placement.
		
\end{itemize}

To the best of our knowledge, this work is the most expansive comparison of network topology throughput to date, and the only one to use accurate and consistent methods.  Our results provide insights that are useful in topology selection and workload placement.  But just as importantly, our evaluation framework and the set of topologies tested are freely available~\cite{tool}. We hope our tools will facilitate future work on rigorously designing and evaluating networks, and replication of research results.

\section{Metrics for throughput}
\label{sec:Metrics}

\pbgnewer{In this section, we define throughput precisely.  We then study how to evaluate worst-case throughput: we consider cut-metrics, before presenting a heuristic algorithm that produces a near-worst-case traffic matrix.}

\subsection{Throughput defined}
\label{sec:throughput-defined}

Our focus in this paper is on network throughput. Furthermore, our focus is on \emph{network topology}, not on higher-level design like routing and congestion control.\footnote{There are many performance measures for networks that are not considered here, and higher-layer systems obviously affect throughput experienced by applications.  But our goal is a comprehensive understanding of the fundamental goal of high-throughput network topology. At this stage, systems-level issues would be a distraction.} Therefore, the metric of interest is end-to-end throughput supported by a network in a fluid-flow model with optimal routing. We next define this more precisely.

A \textbf{network} is a graph $G=(V,E)$ with capacities $c(u,v)$ for every edge $(u,v)\in E_G$.   Among the nodes $V$ are \textbf{servers}, which send and receive traffic flows, connected through non-terminal nodes called \textbf{switches}. Each server is connected to one switch, and each switch is connected to zero or more servers, and other switches.  Unless otherwise specified, for switch-to-switch edges $(u,v)$, we set $c(u,v)=1$, while server-to-switch links have infinite capacity.  This allows us to stress-test the network topology itself, rather than the servers.


A \textbf{traffic matrix (TM)} $T$ defines the traffic demand: for any two servers $v$ and $w$, $T(v,w)$ is an amount of requested flow from $v$ to $w$. We assume without loss of generality that the traffic matrix is normalized so that it conforms to the ``hose model'': each server sends \pbgnew{at most $1$ unit of traffic and receives at most $1$ unit of traffic ($\forall v$, $\sum_w T(v,w) \leq 1$ and $\sum_w T(w,v) \leq 1$)}.

The \textbf{throughput} of a network $G$ with TM $T$ is the maximum value $t$ for which $T\cdot t$ is feasible in $G$. That is, we seek the maximum $t$ for which there exists a feasible multicommodity flow that routes flow $T(v,w)\cdot t$ through the network from each $v$ to each $w$, subject to the link capacity and the flow conservation constraints. This can be formulated in a standard way as a linear program (omitted for brevity) and is thus computable in polynomial time. If the nonzero traffic demands $T(v,w)$ are equal, this is equivalent to the \emph{maximum concurrent flow}~\cite{maxconcflow} problem:
maximizing the minimum throughput of any requested end-to-end flow.


Note that we could alternately maximize the \emph{total} throughput of all flows.  We avoid this because it would allow the network to pick and choose among the TM's flows, giving high bandwidth only to the ``easy'' ones (e.g., short flows that do not cross a bottleneck).  The formulation above ensures that \emph{all} flows in the given TM can be satisfied at the desired rates specified by the TM, all scaled by a constant factor.

We now have a precise definition of throughput, but it depends on the choice of TM.  How can we evaluate a \emph{topology} independent of assumptions about the traffic, if we are interested in worst-case traffic? 

\subsection{Cuts: a weak substitute for worst-case throughput}
\label{subsec:wrongMetric}
\label{sec:sparsest-cut-defined}
Cuts are generally used as proxies to estimate throughput. Since any cut in the graph upper-bounds the flow across the cut, if we find the minimum cut, we can bound the worst-case performance. Two commonly used cut metrics are:

\textit{{(a) Bisection bandwidth:} }It is a widely used to provide an evaluation of a topology's performance independent of a specific TM. It is the capacity of the worst-case cut that divides the network into two equal halves (\cite{padua2011encyclopedia}, p. 974).

\textit {(b) Sparsest cut:} \pbgnewer{The sparsity of a cut is the ratio of its capacity to the net weight of flows that traverse the cut, where the flows depend on a particular TM. Sparsest cut refers to the minimum sparsity in the network.  The special case of \emph{uniform sparsest cut} assumes the all-to-all TM.}

\pbgnewer{Cuts provide an upper-bound on worst-case network performance, are simple to state, and can sometimes be calculated with a formula.  However, they have several limitations.}

\paragraphb{\pbgnewer{(1)} Sparsest cut and bisection bandwidth are not actually TM-independent}, contrary to the original goal of evaluating a topology independent of traffic.  Bisection bandwidth and the uniform sparsest cut correspond to the worst cuts for the complete (all-to-all) TM, so they have a hidden implicit assumption of this particular TM.

\paragraphb{\pbgnewer{(2)} Even for a specific TM, computing cuts is NP-hard}, and it is believed that there is no efficient constant factor approximation algorithm~\cite{chawla2006hardness, bisectionNPComplete}. In contrast, throughput is computable in polynomial time for any specified TM. 

\paragraphb{(3) Cuts are only a \emph{loose} upper-bound for \pbgnewer{worst-case} throughput.}  This may be counter-intuitive if our intuition is guided by the well-known max-flow min-cut theorem which states that in a network with a single flow, the maximum achievable flow is equal to the minimum capacity over all cuts separating the source and the destination~\cite{max-min1,max-min2}. \pbgnew{However, this no longer holds when} there are more than two flows in the network, i.e., \emph{multi-commodity flow}: the maximum flow (throughput) can be an $O(\log n)$ factor lower than the sparsest cut~\cite{max-min-bound}. Hence, cuts do not directly capture the maximum flow.

\begin{figure}
\centering
\includegraphics[width=3in]{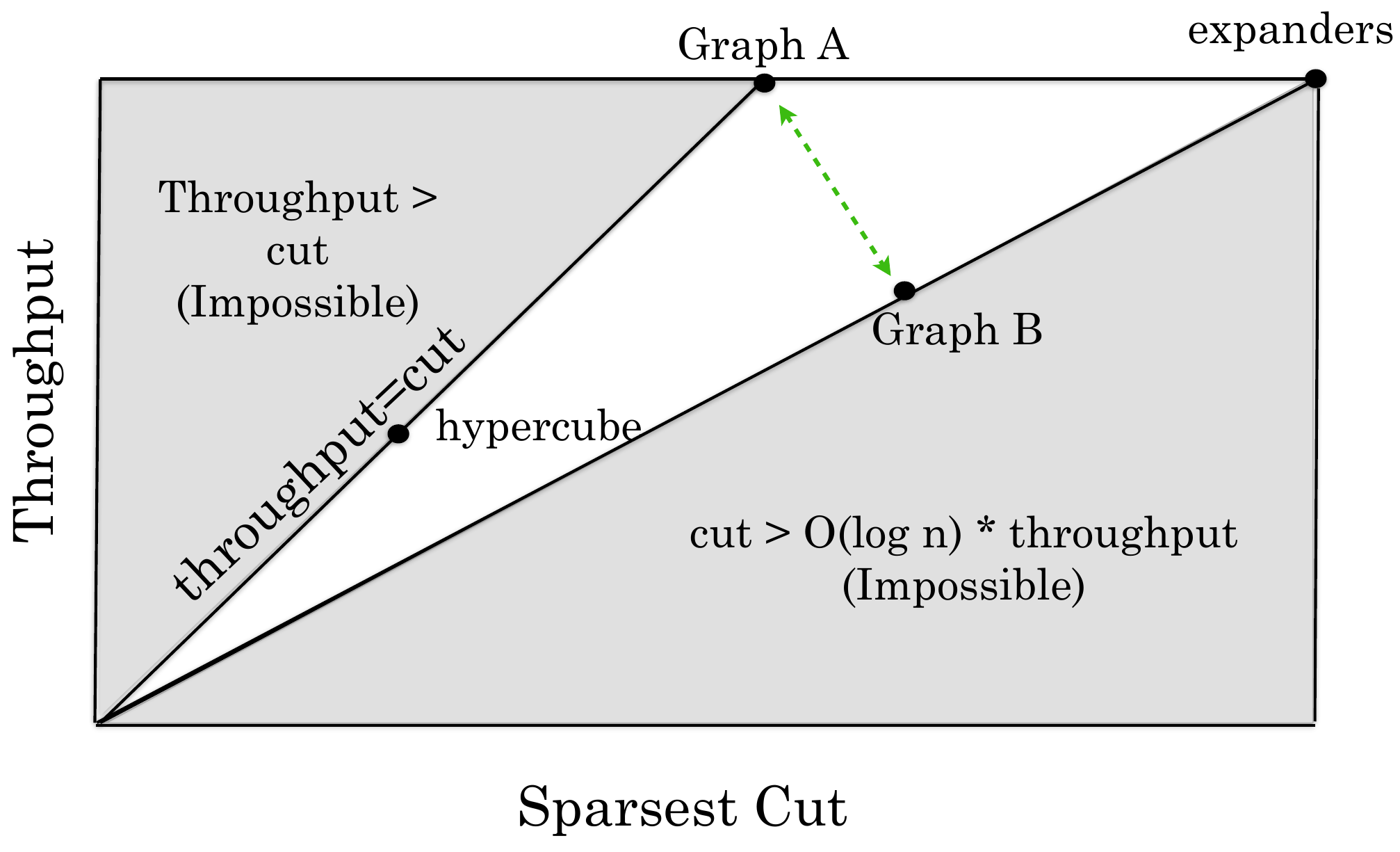}
\caption{Sparsest cut vs. Throughput }
\label{fig:cutVs}
\end{figure}

\pbgnewer{Figure~\ref{fig:cutVs} depicts this relationship between cuts and throughput.  Here we focus on sparsest cut.}\footnote{\pbgnewer{We pick one for simplicity, and sparsest cut has an advantage in robustness.  Bisection bandwidth is forced to cut the network in equal halves, so it can miss more constrained bottlenecks that cut a smaller fraction of the network.}} The flow (throughput) in the network cannot exceed the upper bound imposed by the worst-case cut. On the other hand, the cut cannot be more than a factor $O(\log n)$ greater than the flow~\cite{max-min-bound}. Thus, any graph and an associated TM can be represented by a unique point in the region bounded by these limits.

\SC{While this distinction is well-established~\cite{max-min-bound}, we strengthen the point by showing} that \emph{it can lead to incorrect decisions when evaluating networks}.  Specifically, we will exhibit a pair of graphs $A$ and $B$ such that, as shown in Figure~\ref{fig:cutVs}, $A$ has higher throughput  but $B$ has higher sparsest cut.  If sparsest cut is the metric used to choose a network, graph $B$ will be wrongly assessed to have better performance than graph $A$, while in fact it has a factor $\Omega(\sqrt{\log n})$ worse performance! 

\textbf{Graph A} is a clustered random graph adapted from previous work~\cite{singla14throughput} \pbgnew{with $n$} nodes and degree $2d$. $A$ is composed of two equal-sized clusters with $n/2$ nodes each. Each node in a cluster is connected by degree $\alpha$ to nodes inside its cluster, and degree $\beta$ to nodes in the other cluster, such that $\alpha + \beta = 2d$. $A$ is sampled uniformly at random from the space of all graphs satisfying these constraints. We can pick $\alpha$ and $\beta$ such that $\beta = \Theta(\frac{\alpha}{\log n})$. Then, as per \cite{singla14throughput} (Lemma~$3$), the throughput of $A$ \SC{(denoted $T_A$)} and its sparsest cut \SC{(denoted $\Phi_A$)} are both $\Theta(\frac{1}{n\log n})$.

Let graph $G$ be any $2d$-regular expander on $N=\frac{n}{dp}$ nodes, where $d$ is a constant and $p$ is a parameter we shall adjust later. \textbf{Graph B} is constructed by replacing each edge of $G$ with a path of length $p$. It is easy to see that $B$ has $n$ nodes. We prove in Appendix A the following theorem.

\begin{theorem}\label{thm:graphB}
$T_B = O(\frac{1}{np{\log n}})$ and $\Phi_B=\Omega(\frac{1}{np})$. 
\end{theorem}

In the above, setting $p=1$ corresponds to the `expanders' point in Figure~\ref{fig:cutVs}: both $A$ and $B$ have the same throughput (within constant factors), but the $B$'s cut is larger by $O(\log n)$. Increasing $p$ creates an asymptotic separation in both the cut and the throughput such that $\Phi_A < \Phi_B$, while $T_A > T_B$.

\paragraphb{Intuition.} The reason that throughput may be smaller than sparsest cut is that in addition to being limited by bottlenecks, the network is limited by the total volume of ``work'' it has to accomplish within its total link capacity.  That is, if the TM has equal-weight flows,
\[ 
Throughput\: per \: flow \leq \frac{Total\: link\: capacity}{\#\: of\: flows\: \cdot \: Avg\: path\: length}
\]
where the total capacity is $\sum_{(i,j) \in E} c(i,j)$ and the average path length is computed over the flows in the TM.  This ``volumetric'' upper bound may be tighter than a cut-based bound. 

\cut{
\begin{figure*}
\centering
\begin{minipage}[b]{0.45\linewidth}
	\includegraphics[width=3.1in, height=1.75in]{figures/BGraph.pdf}
	\centering
	\small Graph A - Fat tree
\end{minipage}
\begin{minipage}[b]{0.4\linewidth}
	\includegraphics[width=3.1in, height=1.75in]{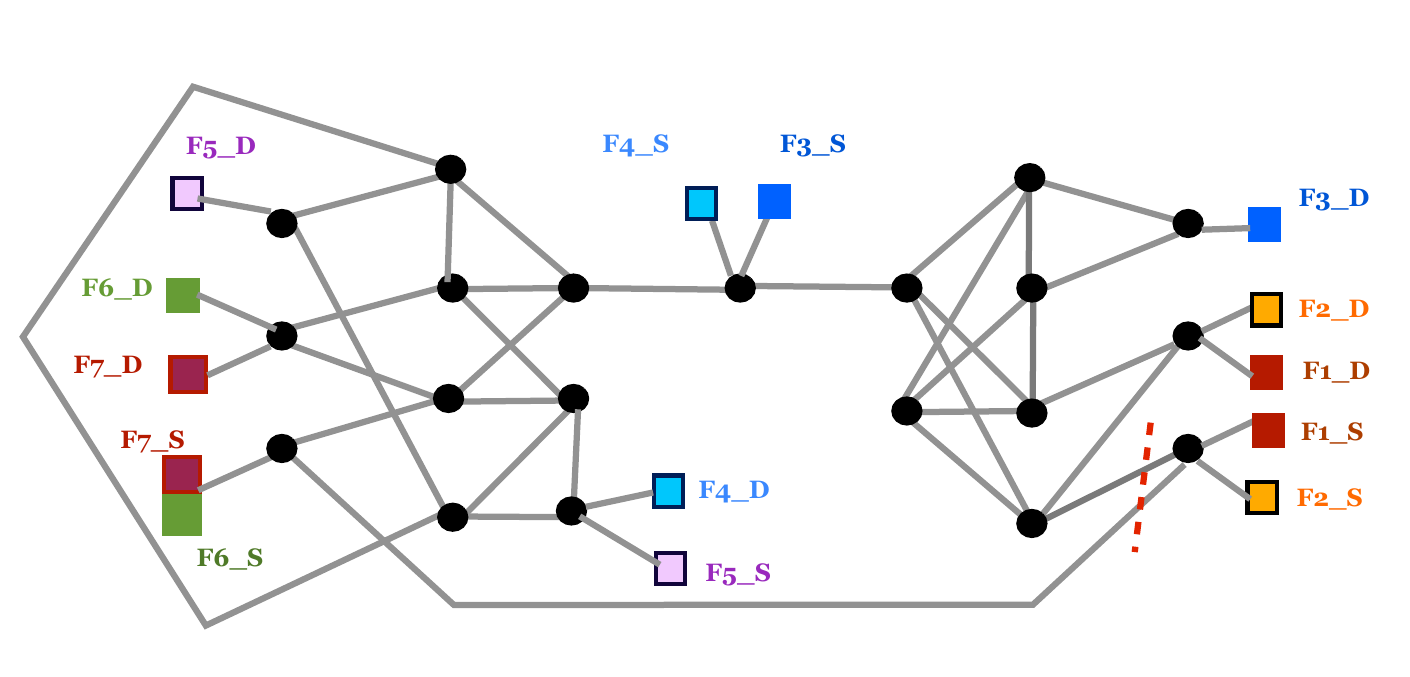}
	\centering
	\small Graph B - Random Graph
\end{minipage}
\caption{\small Fat tree and random graph with relevant source-destination pairs and sparsest-cut}
\label{fig:ftCompare}
\end{figure*}
}


\pbgnewer{Based on the above result and experimental confirmation in \S\ref{sec:results}, we conclude that cuts are not a good choice to evaluate throughput.  Cuts still capture important topological properties such as the difficulty of partitioning a network, and physical layout (packaging and cabling) complexity.  However, these properties are outside the scope of our study of throughput.}


\subsection{Towards a worst-case throughput metric}
\label{subsec:thputMetric}

\begin{figure*}
\centering
\begin{minipage}[b]{0.325\linewidth}
	\includegraphics[width=2.4in]{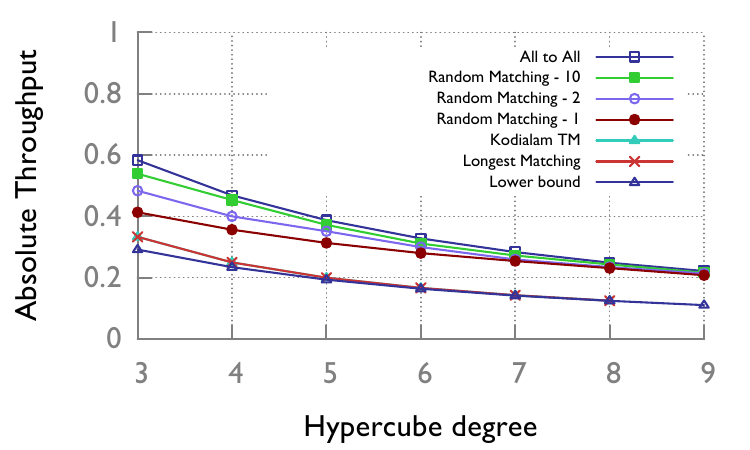}
	\centering
\end{minipage}
\begin{minipage}[b]{0.325\linewidth}
	\includegraphics[width=2.4in]{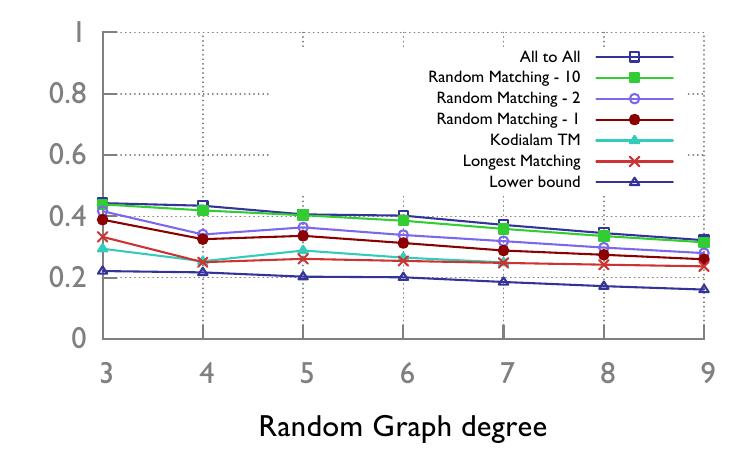}
	\centering
\end{minipage}
\begin{minipage}[b]{0.3\linewidth}
	\includegraphics[width=2.4in]{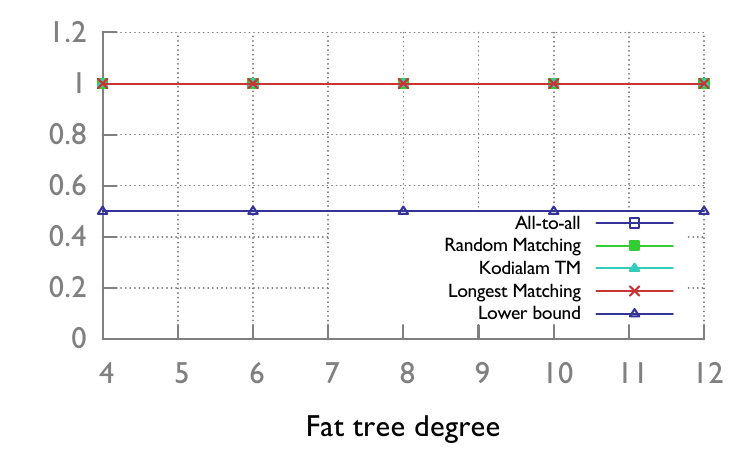}
	\centering
\end{minipage}
\caption{\small Throughput resulting from several different traffic matrices in three topologies}
\label{fig:TCompare}
\end{figure*}

\pbgnewer{Having exhausted cut-based metrics, we return to the original metric of throughput defined in \S\ref{sec:throughput-defined}.  We can evaluate network topologies directly in terms of throughput (via LP optimization software) for specific TMs.  The key, of course, is how to choose the TM.  Our evaluation can include a variety of synthetic and real-world TMs, but as discussed in the introduction, we also want to evaluate topologies' robustness to \emph{unexpected} TMs.}

\pbgnewer{If we can find a worst-case TM, this would fulfill our goal.} However, computing a worst-case TM is an unsolved, computationally non-trivial problem~\cite{hardness}.\footnote{Our problem corresponds to the separation problem of the minimum-cost robust network design in~\cite{hardness}. This problem is shown to be hard for the single-source hose model. However, the complexity is unknown for the hose model with multiple sources which is the scenario we consider.}
Here, we offer an efficient heuristic to find a \emph{near-worst-case TM} which can be used to benchmark topologies.

We begin with the complete or \textbf{all-to-all TM} $T_{A2A}$ which for all $v,w$ has $T_{A2A}(v,w)=\frac{1}{n}$. We observe that $T_{A2A}$ is within $2\times$ of the worst case TM. This fact is simple and known to some researchers, but at the same time, we have not seen it in the literature, so we give the statement here and proof in Appendix B.

\begin{theorem}\label{thm:a2a}
Let $G$ be any graph.  If $T_{A2A}$ is feasible in $G$ with throughput $t$, then any hose model traffic matrix is feasible in $G$ with throughput $\geq t/2$.
\end{theorem}

Can we get closer to the worst case TM?  In our experience, TMs with a smaller number of ``elephant'' flows are more difficult to route than TMs with a large number of small flows, like $T_{A2A}$.  This suggests a \textbf{random matching TM} in which we have only one outgoing flow and one incoming flow per server, chosen uniform-randomly among servers. 


Can we get \emph{even closer} to the worst-case TM?  Intuitively, the all-to-all and random matching TMs will tend to find \pbgnew{sparse} cuts, but only have average-length paths.  Drawing on the intuition that throughput decreases with average flow path length, we seek to produce traffic matrices that force the use of long paths.  To do this, given a network $G$, we compute all-pairs shortest paths and create a complete bipartite graph $H$, whose nodes represent all sources and destinations in $G$, and for which the weight of edge $v\to w$ is the length of the shortest $v\to w$ path in $G$.  We then find the maximum weight matching in $H$. The resulting matching corresponds to the pairing of servers which maximizes average flow path length, assuming flow is routed on shortest paths between each pair.  We call this a \textbf{longest matching TM}, \pbgnewer{and it will serve as our heuristic for a near-worst-case traffic.}

\SC{Kodialam et al.~\cite{kodialam} proposed another heuristic to find a near-worst-case TM:  maximizing the average path length of a flow.  This \textbf{Kodialam TM} is similar to the longest matching but may have many flows attached to each source and destination.  This TM was used in~\cite{kodialam} to evaluate oblivious routing algorithms, but there was no evaluation of how close it is to the worst case, so our evaluation here is new.}

Figure~\ref{fig:TCompare} shows the resulting throughput of these TMs in three topologies: hypercubes, random regular graphs, and fat trees~\cite{fattree-new}. In all cases, A2A traffic has the highest throughput; throughput decreases or does not change as we move to a random matching TM with $10$ servers per switch, and progressively decreases as the number of servers per switch is decreased to $1$ under random matching, and finally to the Kodialam TM and the longest matching TM.  We also plot \pbgnew{the lower bound given by Theorem~\ref{thm:a2a}}: $T_{A2A}/2$. Comparison across topologies is not relevant here since the topologies are not built with the same ``equipment'' (node degree, number of servers, etc.)

We chose these three topologies to illustrate cases when our approach is most helpful, somewhat helpful, and least helpful at finding near-worst-case TMs.  In the \textbf{hypercube}, the longest matching TM is extremely close to the worst-case performance. To see why, note that the longest paths have length $d$ in a $d$-dimensional hypercube, twice as long as the mean path length; and the hypercube has $n\cdot d$ uni-directional links. The total flow in the network will thus be $(\#\: flows \cdot average\: flow\: path\: length) = n\cdot d$. \pbgnew{Thus,} all links will be perfectly utilized if the flow can be routed, which empirically it is.  In the \textbf{random graph}, there is less variation in end-to-end path length, but across our experiments the longest matching is always within $1.5\times$ of the provable lower bound (\pbgnew{and may be even closer to the true lower bound, since Theorem~\ref{thm:a2a}} may not be tight).  In the \textbf{fat tree}, which is here built as a three-level nonblocking topology, there is essentially no variation in path length since asymptotically nearly all paths go three hops up to the core switches and three hops down to the destination.  Here, our TMs are no worse than all-to-all, and the simple $T_{A2A}/2$ lower bound is off by a factor of $2$ from the true worst case (which is throughput of $1$ as this is a nonblocking topology).

The longest matching and Kodialam TMs are identical in hypercubes and fat trees. On random graphs, they yield slightly different TMs, with longest matching yielding marginally better results at larger sizes.  In addition, longest matching has a significant practical advantage:  it produces far fewer end-to-end flows than the Kodialam TM.  Since the memory requirements of computing throughput of a given TM (via the multicommodity flow LP) depends on the number of flows, longest matching requires less memory and compute time. For example, in random graphs on a 32 GB machine using the Gurobi optimization package, the Kodialam TM can be computed up to 128 nodes while the longest matching scales to 1,024, \pbgnewer{while being computed roughly $6\times$ faster}. Hence, we choose longest matching as our representative near-worst-case traffic matrix.



\subsection{Summary and implications}




Directly evaluating throughput with particular TMs using LP optimization is both more accurate and more tractable than cut-based metrics.  In choosing a TM to evaluate, both ``average case'' and near-worst-case TMs are reasonable choices. \pbgnewer{Our evaluation will employ multiple synthetic TMs and measurements from real-world applications. For near-worst-case traffic, we developed a practical heuristic, the longest matching TM, that often succeeds in substantially worsening the TM compared with A2A.}

\pbgnewer{Note that measuring throughput directly (rather than via cuts) is not in itself a novel idea: numerous papers have evaluated particular topologies on particular TMs.  Our contribution is to provide a rigorous analysis of why cuts do not always predict throughput; a way to generate a near-worst-case TM for any given topology; and an experimental evaluation benchmarking a large number of proposed topologies on a variety of TMs.}

\section{Experimental Methods}
\label{sec:results}

In this section, we present our experimental methodology and detailed analysis of our framework. We attempt to answer the following questions: Are cut-metrics indeed worse predictors of performance?  When measuring throughput directly, how close do we come to worst-case traffic?

\subsection{Methodology}
Before delving into the experiments, we explain the methods used for computing throughput. We also elaborate on the traffic workloads and topologies used in the experiments.
%
%
%


\subsubsection{\pbgnew{Computing throughput}}

Throughput is computed as a solution to a linear program whose objective is to maximize the minimum flow across all flow demands, as explained in \S\ref{sec:Metrics}. We use the Gurobi~\cite{gurobi} linear program solver. \pbgnew{Throughput} depends on the traffic workload provided to the network.

\subsubsection{Traffic workload} 
We evaluate two main categories of workloads: (a) real-world measured TMs from Facebook clusters and (b) synthetic TMs, which can be uniform weight or non-uniform weight. Synthetic workloads belong to three main families: all-to-all, random matching and longest matching (near-worst-case). In addition, we need to specify where the traffic endpoints (sources and destinations, i.e., servers) are attached. In structured networks with restrictions on server-locations (fat-tree, BCube, DCell), servers are added at the locations prescribed by the models. For example, in fat-trees, servers are attached only to the lowest layer. For all other networks, we add servers to each switch. Note that our traffic matrices effectively encode switch-to-switch traffic, so the particular \emph{number} of servers doesn't matter.

\subsubsection{Topologies}

Our evaluation uses $10$ families of computer networks. Topology families evaluated are: BCube~\cite{BCube}, DCell~\cite{dcell}, Dragonfly~\cite{dragonfly}, Fat Tree~\cite{fatTree}, Flattened butterfly~\cite{flatBF}, Hypercubes~\cite{hypercubenetwork}, HyperX~\cite{hyperx}, Jellyfish~\cite{jellyfish}, Long Hop~\cite{ANCS} and Slim Fly~\cite{slim-fly}. For evaluating the cut-based metrics in a wider variety of environments, we consider $66$ non-computer networks -- food webs, social networks, and more~\cite{tool}.

\subsection{\pbgnewer{Do cuts predict worst-case throughput?}}

\pbgnewer{In this section, we experimentally evaluate whether cut-based metrics predict throughput.} We generate multiple networks from each of our topology families (with varying parameters such as size and node degree), compute throughput with the longest matching TM, and find sparse cuts using heuristics with the same longest matching TM. We show that:

\begin{itemize}[leftmargin=*]


	\item In several networks, bisection bandwidth cannot predict throughput accurately. For a majority of large networks, our best estimate of sparsest-cut differs from the computed worst-case throughput by a large factor.

	\item Even in a well-structured network of small size (where brute force is feasible), sparsest-cut can be higher than worst-case throughput.

	
	
\end{itemize}


\begin{figure*}[ht!]
\begin{minipage}[b]{0.55\linewidth}
	\includegraphics[width=3.6in]{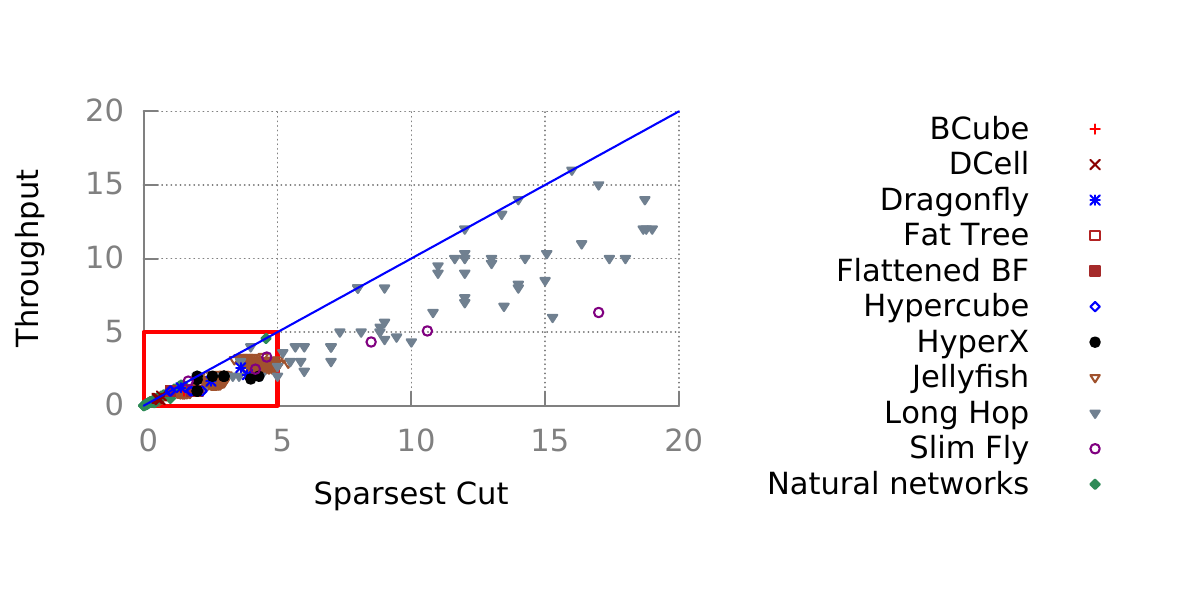}
	\newline (a) Throughput vs. cut for all graphs
	
\end{minipage}
\begin{minipage}[b]{0.3\linewidth}
	\includegraphics[width=2.5in]{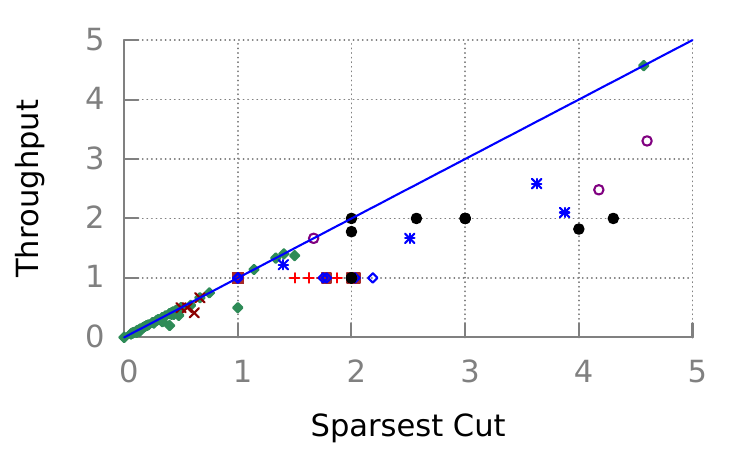}
	\newline (b) Throughput vs. cut for selected graphs \SC{(zoomed version of (a))}
	\centering
\end{minipage}
\caption{\small Throughput vs. cut.  (Comparison is valid for individual networks, not \emph{across} networks, since they have different numbers of nodes and degree distributions.)}
\label{fig:cutvFlow}
\end{figure*}

Since brute-force computation of cuts is possible only on small networks (up to 20 nodes), we evaluate bisection bandwidth and sparsest cut on networks of feasible size ($115$ networks total -- $100$ Jellyfish networks and $15$ networks from $7$ families of topologies). Of the $8$ topology families tested, we found that bisection bandwidth accurately predicted throughput in only $5$ of the families while sparsest cut gives the correct value in $7$. The average error (difference between cut and throughput) is $7.6\%$ for bisection bandwidth and $6.2\%$ for sparsest cut (in networks where they differ from throughput). Maximum error observed was $56.3\%$ for bisection bandwidth and $6.2\%$ for sparsest cut.

Although sparsest cut does a better job at estimating throughput at smaller sizes, we have found that in a $5$-ary $3$-stage flattened butterfly with only $25$ switches and $125$ servers, the throughput is less than the sparsest cut (and the bisection bandwidth). Specifically, the absolute throughput in the network is $0.565$ whereas the sparsest-cut is $0.6$. This shows that even in small networks, throughput can be different from the worst-case cut.  While the differences are not large in the small networks where brute force computation is feasible, note that since cuts and flows are separated by an $\Theta(\log n)$ factor, we should expect them to grow.

Sparsest cut being the more accurate cut metric, we extend the sparsest cut estimation to larger networks.  Here we have to use heuristics to compute sparsest cut, \pbgnewer{but we compute all of an extensive set of heuristics} (limited brute-force, eigenvector-based optimization, etc. detailed in Appendix C) and use the term \textbf{sparse cut} to refer to the sparsest cut that was found by any of the heuristics. Sparse cuts differ from throughput substantially, with up to a $3\times$ discrepancy as shown in Figure~\ref{fig:cutvFlow}.  In only a small number of cases, the cut equals throughput.  The difference between cut and throughput is pronounced. For example, Figure~\ref{fig:cutvFlow}(b) shows that although HyperX networks of different sizes have approximately same flow value (y axis), they differ widely in sparsest cut (x axis). This shows that estimation of worst-case throughput performance of networks using cut-based metrics can lead to erroneous results.

\begin{figure}[t]
\centering
\includegraphics[width=3in]{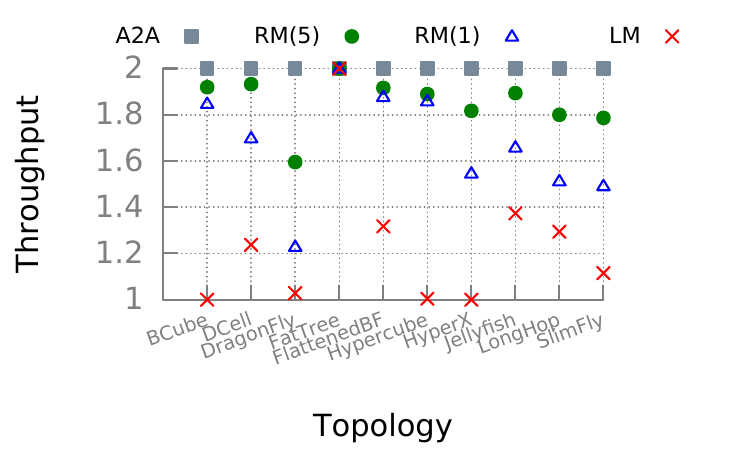}
\caption{\small Comparison of throughput under different traffic matrices normalized with respect to theoretical lower bound}
\label{fig:compareTMsAll}
\end{figure}

\subsection{\pbgnewer{Does longest matching approach the true worst case?}} 


We compare representative samples from each family of topology under \pbgnew{four types of} TM: all to all (A2A), random matching with $5$ servers per switch, random matching with $1$ server per switch, and longest matching. Figure~\ref{fig:compareTMsAll} shows the throughput values normalized so that the theoretical lower bound on throughput is $1$, and therefore A2A's throughput is $2$. For all networks, $ T_{A2A} \geq T_{RM(5)} \geq T_{RM(1)} \geq T_{LM} \geq 1$, \pbgnewer{i.e., all-to-all is consistently the easiest TM in this set, followed by random matchings, longest matching, and the theoretical lower bound.  This confirms our intuition discussed in \S\ref{subsec:thputMetric}.}  (As in Figure \ref{fig:cutvFlow}, throughput comparisons are valid across TMs for a particular network, not across networks since the number of links and servers varies across networks.)


Our longest matching TM is successful in matching the lower bound for BCube, Hypercube, HyperX, and (nearly) Dragonfly.  In all other families except fat trees, the traffic under longest matching is significantly closer to the lower bound than with the other TMs.  In fat trees, throughput under A2A and longest matching are equal. This is not a shortcoming of the metric, rather it’s the lower bound which is loose here: in fat trees, it can be easily verified that the flow out of each top-of-rack switch is the same under all symmetric TMs (i.e., with equal-weight flows and the same number of flows into and out of each top-of-rack switch).


\pbgnewer{In short, these results show that (1) the heuristic for near-worst-case traffic, the longest matching TM, is a significantly more difficult TM than A2A and RM and often approaches the lower bound; and (2) throughput measurement using longest matching is a more accurate estimate of worst-case throughput performance than cut-based approximations, in addition to being substantially easier to compute.}




\section{Topology evaluation}
\label{sec:expResults}

In this section, we present the results of our topology evaluation with synthetic and real-world workloads\pbgnewer{, and our near-worst-case TM.}

\pbgnewer{But first,} there is one more piece of the puzzle to allow comparison of networks. Networks may be built with different equipment -- with a wide variation in number of switches and links. The raw throughput value does not account for this difference in hardware resources, \pbgnew{and most proposed topologies can only be built with particular discrete numbers of servers, switches, and links, which inconveniently do not match.}

\begin{figure}[ht]
\centering
\begin{minipage}[b]{\linewidth}
	\includegraphics[width=2.2in]{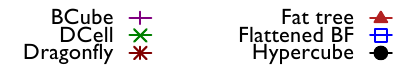}
	\centering
\end{minipage}
\begin{minipage}[b]{\linewidth}
	\includegraphics[width=2.8in]{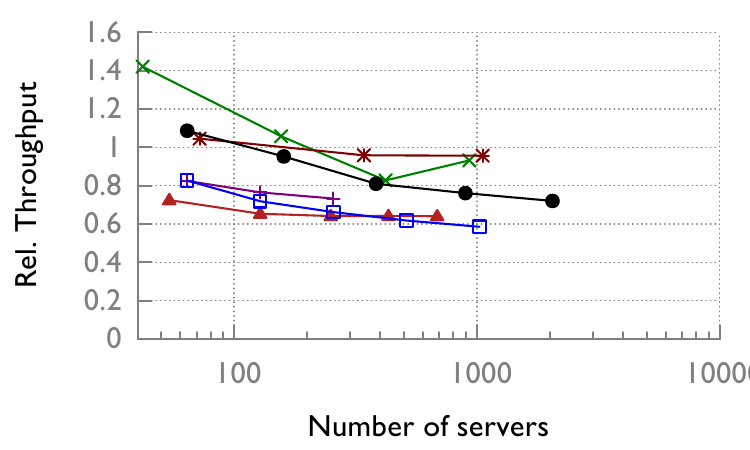}
	\newline (a) All to All TM
	\centering
\end{minipage}
\begin{minipage}[b]{\linewidth}
	\includegraphics[width=2.8in]{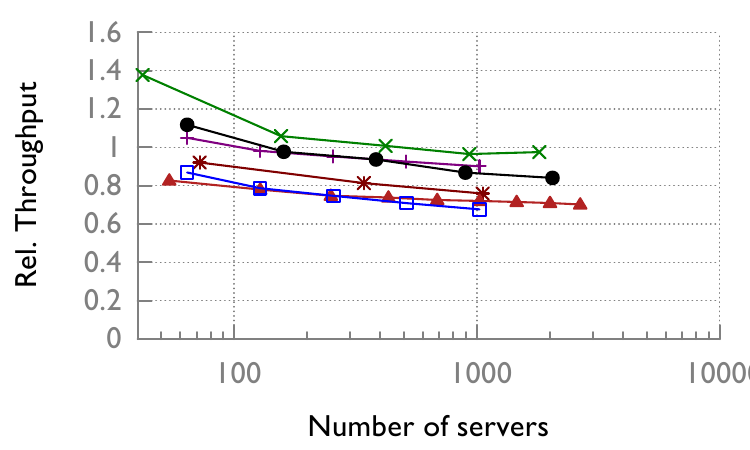}\\
	\newline (b) Random Matching TM
	\centering
	\label{fig:RandTM}
\end{minipage}
\begin{minipage}[b]{\linewidth}
	\includegraphics[width=2.8in]{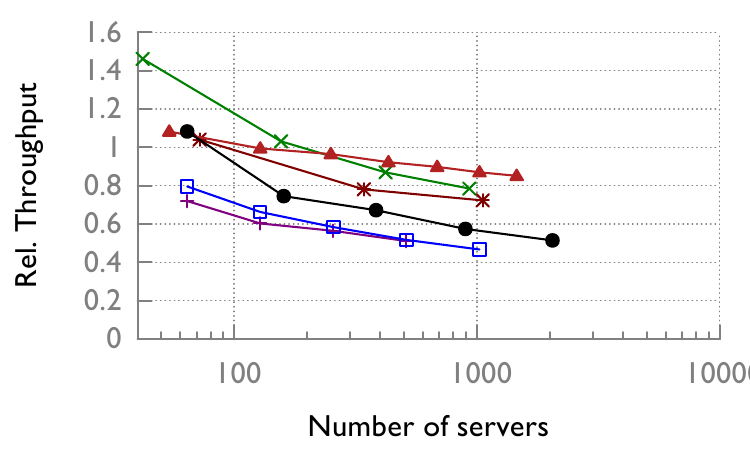}
	\newline (c) Longest Matching TM
	\centering
\end{minipage}
\caption{\small Comparison of TMs on topologies}
\label{fig:AllCompare1}
\end{figure}

\begin{figure}[ht]
\centering
\begin{minipage}[b]{\linewidth}
	\includegraphics[width=2in]{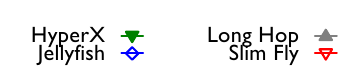}
	\centering
\end{minipage}
\begin{minipage}[b]{\linewidth}
	\includegraphics[width=2.8in]{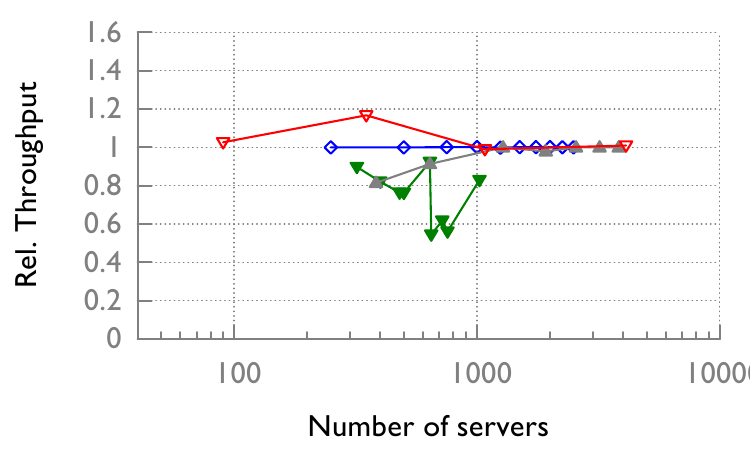}
	\newline (a) All to All TM
	\centering
\end{minipage}
\begin{minipage}[b]{\linewidth}
	\includegraphics[width=2.8in]{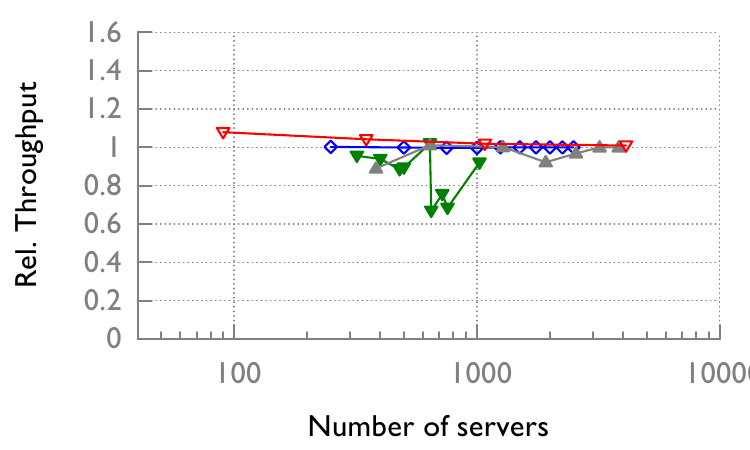}
	\newline (b) Random Matching TM
	\centering
\end{minipage}
\begin{minipage}[b]{\linewidth}
	\includegraphics[width=2.8in]{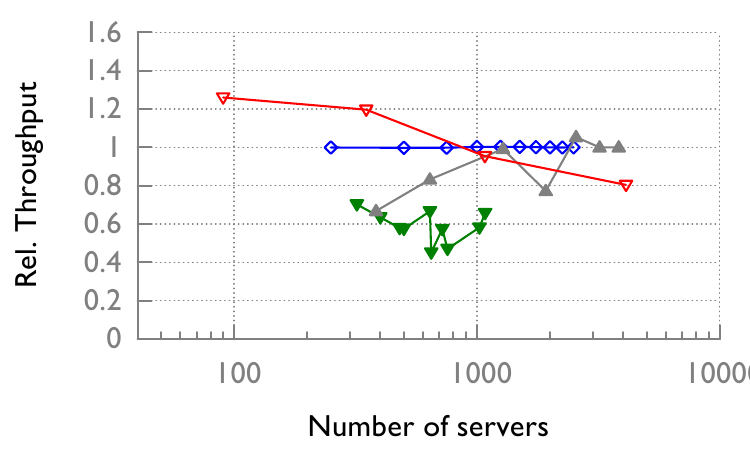}
	\newline (c) Longest Matching TM
	\centering
\end{minipage}
\caption{\small Comparison of TMs on \pbgnew{more} topologies}
\label{fig:AllCompare2}
\end{figure}

\pbgnew{Fortunately, uniform-random graphs as in~\cite{jellyfish} can be constructed for any size and degree distribution.  Hence, random graphs serve as a convenient benchmark for easy comparison of network topologies.}  Our high-level approach to evaluating a network is to: (i) compute the network's throughput; (ii) build a random graph with precisely the \pbgnew{same equipment, i.e., the same number of nodes each with the same number of links as the corresponding node in the original graph, (iii)} compute the throughput of \pbgnew{this same-equipment random graph} under the same TM; (iv) normalize the network's throughput with the random graph's throughput for comparison against other \pbgnew{networks. This} normalized value is referred to as {\bf relative throughput}.  Unless otherwise stated, each data-point is an average across $10$ iterations, and all error bars are $95\%$ two-sided confidence intervals. \ajsnew{Minimal variations lead to narrow error bars in networks of size greater than $100$.}




\vspace{-2mm}


\subsection{Synthetic Workloads}

We evaluate the performance of $10$ types of computer networks under two categories of synthetic workloads: (a) uniform weight and (b) non-uniform weight.


\paragraphbi{1) Uniform-Weight Synthetic Workloads}

We evaluate three traffic matrices with equal weight across flows: all to all, random matching with one server, and longest matching.
Figures~\ref{fig:AllCompare1} and \ref{fig:AllCompare2} show the results, divided among two figures for visual clarity.  Most topologies are in Figure~\ref{fig:AllCompare1}, while Figure~\ref{fig:AllCompare2} shows Jellyfish, Slim Fly, Long Hop, and HyperX.

\paragraphb{Overall,} \pbgnew{performance varies substantially, by around $1.6\times$ with A2A traffic and more than $2\times$ with longest matching.  Furthermore, for the networks of Figure~\ref{fig:AllCompare1}, which topology ``wins'' depends on the TM.
For example, Dragonfly has high relative throughput under A2A but under the longest matching TM, fat trees achieve the highest throughput at the largest scale.  However, in all TMs, Jellyfish achieves highest performance (relative throughput $=1$ by definition), with longest matching  at the largest scale (1000+ servers).} In comparison with random graphs, performance degrades for most networks with increasing size.  The relative performance of the largest network tested in each family in Figure~\ref{fig:AllCompare1} is given in Table~\ref{table:relative}. \pbgnewer{The pattern of performance of the networks shown in Figure~\ref{fig:AllCompare2} is different, so we discuss these in detail next.}



\begin{table}[ht]
\small
\centering
  \begin{tabular}{ | L{2cm} | C{1cm} | C{1.3cm} | C{1.3cm} |}
    \hline 
    Topology family & All-To-All & Random Matching &  Longest matching \\ \hline
    BCube ($2$-ary)  & 73\% & 90\% & 51\%   \\ \hline
    DCell ($5$-ary) & 93\% & 97\% & 79\%   \\ \hline
	Dragonfly & 95\% & 76\% & 72\%    \\ \hline
	Fat tree & 65\% & 73\% & 89\%    \\ \hline 
	Flattened BF ($2$-ary) & 59\% & 71\% & 47\%  \\ \hline   	
    	Hypercube & 72\% & 84\% & 51\%  \\ \hline
    \end{tabular}
    \label{table:relative}
    \caption{Relative throughput at the largest size tested under different TMs}
\end{table} 

\begin{figure*}[t]
\hspace{-0.5cm}
  \begin{minipage}[t]{0.02\linewidth}~
  \end{minipage}
      \begin{minipage}[t]{0.3\linewidth}
    \centering
    \includegraphics[scale=.45]{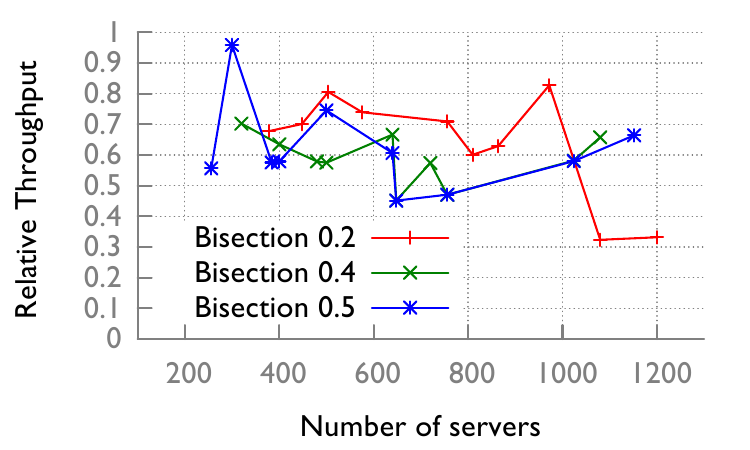}\\
    \caption{\small HyperX Relative Throughput under Longest Matching with different Bisection}
   \label{fig:hyperx}
  \end{minipage}
    \begin{minipage}[t]{0.02\linewidth}~
  \end{minipage}
\begin{minipage}[t]{0.3\linewidth}
    \centering
    \includegraphics[scale=.45]{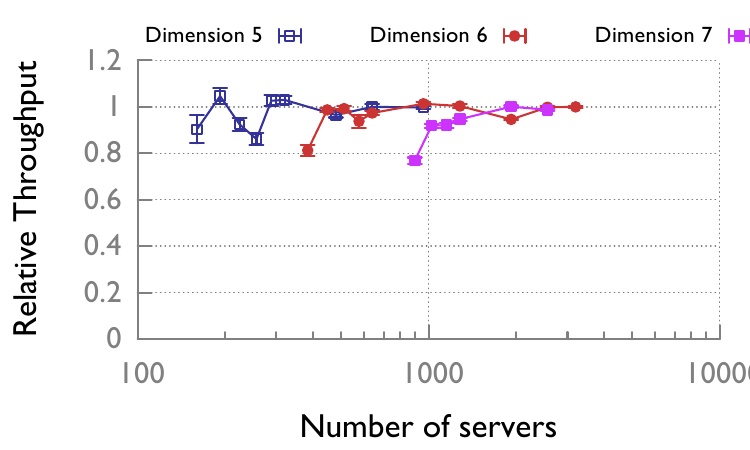}\\
    \caption{\small Long Hop network relative throughput under longest matching TM }
     \label{fig:ANCS}
  \end{minipage}
    \begin{minipage}[t]{0.02\linewidth}~
  \end{minipage}
    \begin{minipage}[t]{0.3\linewidth}
    \centering
    \includegraphics[scale=.45]{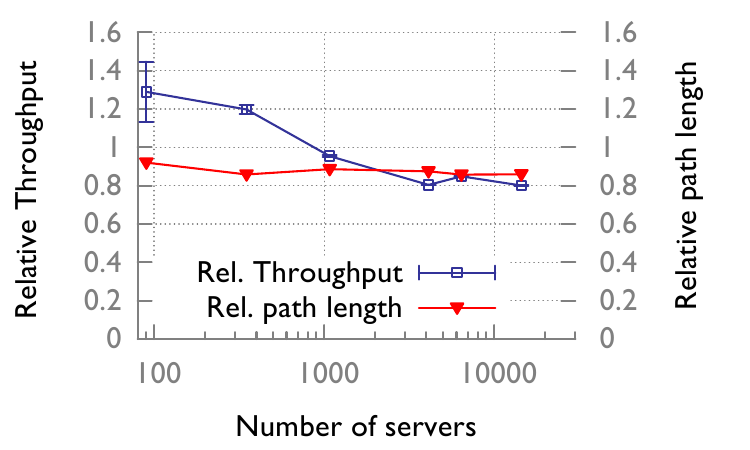}\\
    \caption{ \small Slim Fly Relative Throughput under Longest Matching TM}
\label{fig:slimFly}
  \end{minipage}
\end{figure*}


%

\paragraphb{HyperX}~\cite{hyperx} \pbgnewer{has irregular performance across scale.} The variations arise from the choice of underlying topology: Given a switch radix, number of servers and desired bisection bandwidth, HyperX attempts to find the least cost topology constructed from the building blocks -- hypercube and flattened butterfly. Even a slight variation in one of the parameters can lead to a significant difference in HyperX construction and hence throughput.

We plot HyperX results for networks with bisection bandwidth of $0.4$, i.e., in the worst case the network can support flow from $40\%$ of the hosts without congestion. In the worst case, the HyperX network that can support $648$ hosts achieves a relative throughput of $55\%$ under all-to-all traffic. For the same network, relative performance improves to $67\%$ under random matching but further degrades to $45\%$ under longest matching. We investigate the performance of HyperX networks with different bisection bandwidths under the longest matching TM in Figure~\ref{fig:hyperx} and observe that the performance varies widely with network size under all values of bisection. More importantly, this further illustrates that high bisection does not guarantee high performance.


\paragraphb{Jellyfish}~\cite{jellyfish} \pbgnew{is our benchmark, and its (average) relative performance is $100\%$ under all workloads by definition.  Performance depends on randomness in the construction, but at larger sizes, performance differences between Jellyfish instances are minimal ($\approx 1\%$).}


\paragraphb{Long Hop networks}~\cite{ANCS} In Figure~\ref{fig:ANCS}, we show that the relative throughput of Long Hop networks \pbgnew{approaches} $1$ at large sizes under the longest matching TM. Similar trends are observed under all-to-all and random matching workloads. \pbgnew{The paper~\cite{ANCS} claimed high performance (in terms of bisection bandwidth) with substantially less equipment than past designs, but we find that while Long Hop networks do have high performance, they are no better than random graphs, and sometimes worse.}


\paragraphb{Slim Fly}~\cite{slim-fly}  \pbgnew{It was noted in~\cite{slim-fly} that Slim Fly's key advantage is path length. We observe in Figure~\ref{fig:slimFly} that the topology indeed has extremely short paths -- about $85$-$90\%$ relative to the random graph -- but this does not translate to higher throughput.  Slim Fly's performance is essentially identical to random graphs under all-to-all and random matching TMs with a relative performance of $101\%$ using both.  Relative throughput under the longest matching TM decreases with \pbgnew{scale, dropping} to $80\%$ at the largest size tested. Hence, throughput performance cannot be predicted solely based on the average path length.}

\begin{figure*}[ht!]
\hspace{-0.5cm}
	\begin{minipage}[b]{0.35\linewidth}
		\includegraphics[scale=0.7]{heading1.pdf}
		\centering
	\end{minipage}
		\begin{minipage}[b]{0.35\linewidth}
		\includegraphics[scale=0.7]{heading2.pdf}
		\centering
	\end{minipage}
		\begin{minipage}[b]{0.2\linewidth}
		\includegraphics[scale=0.7]{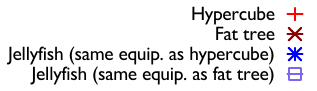}
		\centering
	\end{minipage}
  \newline
  
  \begin{minipage}[b]{0.31\linewidth}
	\includegraphics[scale=0.45]{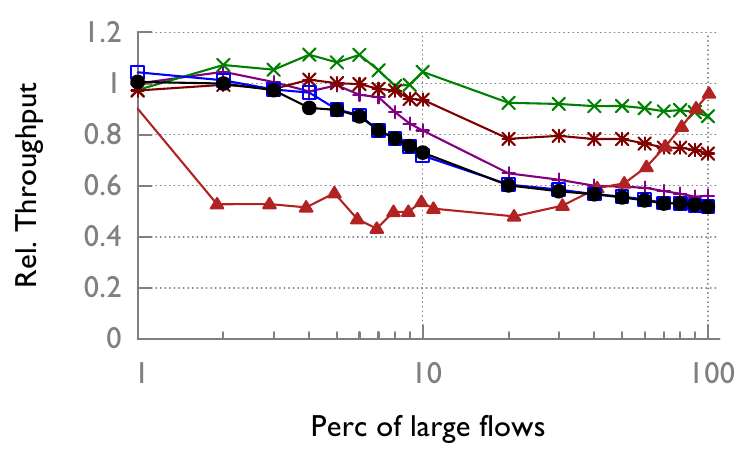}
			\centering
			\caption{\small Comparison of tunable TMs on topologies}
			\label{fig:AllCompareNewTM1}
	\end{minipage}
		    \begin{minipage}[t]{0.02\linewidth}~
  \end{minipage}
  \begin{minipage}[b]{0.31\linewidth}
	\includegraphics[scale=0.45]{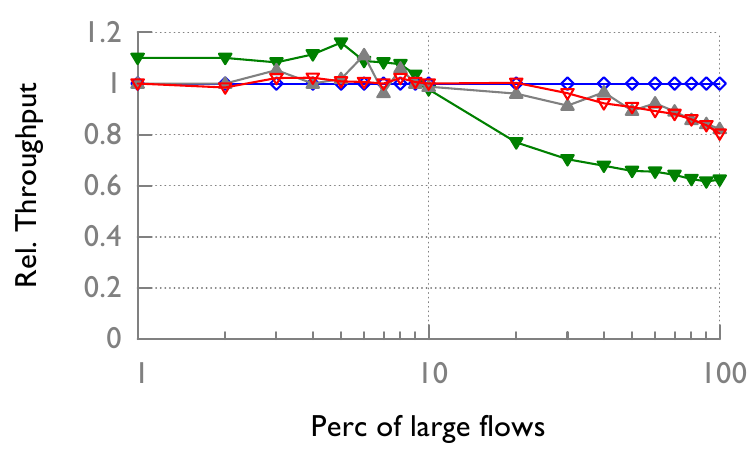}
			\centering
			\caption{\small Comparison of tunable TMs on topologies}
			\label{fig:AllCompareNewTM2}
	\end{minipage}
	    \begin{minipage}[t]{0.02\linewidth}~
  \end{minipage}
	  \begin{minipage}[b]{0.31\linewidth}
	\includegraphics[scale=0.45]{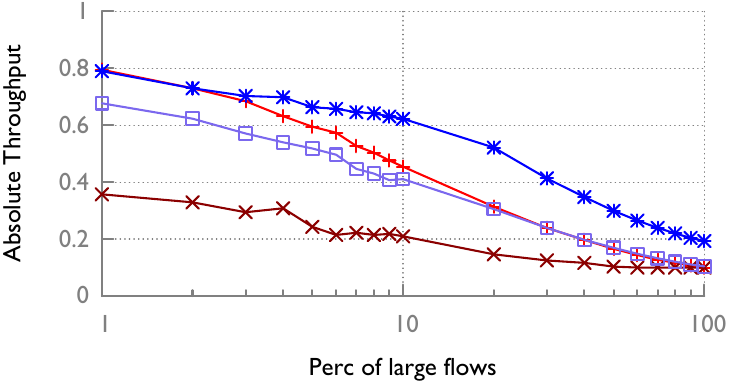}
			\centering
			\caption{\small Comparison of absolute throughput under tunable TMs (Jellyfish network corresponding to same equipment as hypercube and fat tree used for comparison)}
			\label{fig:CompareAbsolutre}
	\end{minipage}
\end{figure*}

\paragraphbi{2) Non-Uniform Weight Synthetic Workloads}

In order to understand the robustness of throughput performance against variations in TM, we consider the response of topologies to non-uniform traffic matrices. Using the longest matching TM, we assign a weight of $10$ to $x\%$ of flows in the network which are chosen randomly while all other flows have a weight of $1$. \pbgnewer{Thus, we have a longest matching TM modified so that a fraction of flows have higher demand than the rest.} We vary $x$ to understand the response of topologies to skewed traffic matrices. The experiment is conducted on a representative sample from each topology class. Note that the relative throughput at $0\%$ (which is not plotted) will be equal to that at $100\%$ since all flows are scaled by the same factor. 

We observe in Figures~\ref{fig:AllCompareNewTM1} and~\ref{fig:AllCompareNewTM2} that all topologies except fat trees perform well when there are a few large flows in the network. To further understand the anomaly with fat trees, we look at the absolute throughput in fat tree, hypercube and Jellyfish network built with matching gear in Figure~\ref{fig:CompareAbsolutre}. Trends of other networks are similar to hypercube. The anomalous behavior of fat trees arises due to the fact that load on links connected to the top-of-rack switch (ToR) in fat trees is from the directly connected servers only. In all other networks, under longest matching, links connected to ToRs carry traffic belonging to flows that do not originate or terminate at the ToR. Since the load is spread more evenly, a few large flows do not affect the overall throughput. Note that fat trees can be built to be non-blocking, with high performance for any traffic; what these results show is that for these TMs, fat trees would pay a high price in equipment cost to achieve full non-blocking throughput, relative to other topology options.

\textbf{Overall, } we find that \pbgnewer{under traffic matrices with non-uniform demands, the Jellyfish, DCell, Long hop and Slim Fly topologies perform better than other topologies}. While all other topologies exhibit graceful degradation of throughput with increase in the percentage of large flows, fat trees give anomalous behavior due to the load on the ToRs.

\subsection{Real-world workloads}


\pbgnewer{Roy et al.~\cite{FBTM} presents relative traffic demands measured during a period of $24$ hours in two $64$-rack clusters operated by Facebook.} The first TM corresponds to a Hadoop cluster and has nearly equal weights. The second TM from a frontend cluster is more non-uniform, with relatively heavy traffic at the cache servers and lighter traffic at the web servers. Since the raw data is not publicly available, we processed the paper's plot images to retrieve the approximate weights for inter-rack traffic demand with an accuracy of $10^i$ in the interval $[10^i, 10^{i+1})$ (from data presented in color-coded log scale). Since our throughput computation \pbgnewer{rescales the TM anyway, relative weights are sufficient}. Given the scale factor, the real throughput can be obtained by multiplying the computed throughput with the scale factor.

We test our slate of topologies with the Hadoop cluster TM (TM-H) and the frontend cluster TM (TM-F) and the results are presented in Figures~\ref{fig:AllCompareFB1} and~\ref{fig:AllCompareFB2} respectively. When a topology family does not have a candidate with $64$ ToRs, the TM is downsampled to the nearest valid size (denoted as Sampled points in the plot).

\pbgnewer{We also consider modifying these TMs by} shuffling the order of ToRs in the TM (denoted as Shuffled points in the plot). Under the nearly uniform TM-H, shuffling does not affect throughput performance significantly (Figure~\ref{fig:AllCompareFB1}). However, under TM-F with skewed weights, shuffling can give significant improvement in performance. In all networks except Jellyfish, Long Hop, Slim Fly and fat trees, randomizing the flows, which in turn spreads the large flows across the network, can improve the performance significantly when the traffic matrix is non-uniform. Note that relative performance here may not match that in Figures~\ref{fig:AllCompare1} and~\ref{fig:AllCompare2} since the networks have different numbers of servers and even ToRs (due to downsampling required for several topologies). However, we observe that the relative performance between topologies remains consistent across both the traffic matrices.

\begin{figure*}[t]
\begin{minipage}[b]{0.31\linewidth}
	\includegraphics[scale=0.45]{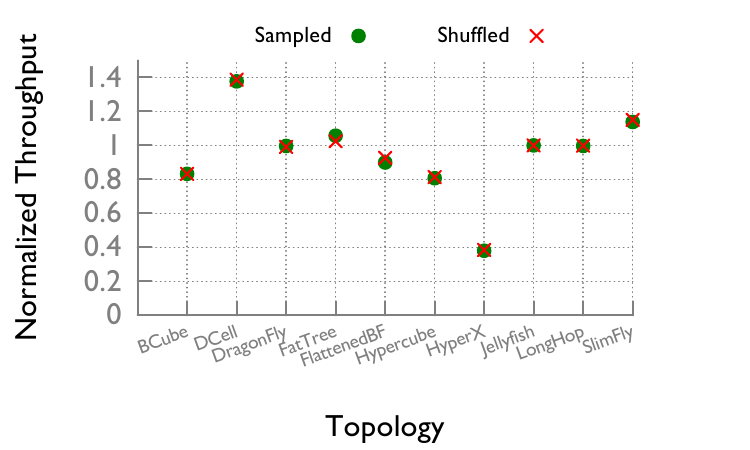}
			\vspace{-2mm}
			\centering
			\caption{\small Comparison of topologies with TM-H(~\cite{FBTM})}
			\label{fig:AllCompareFB1}
	\end{minipage}
		    \begin{minipage}[t]{0.02\linewidth}~
  \end{minipage}
  \begin{minipage}[b]{0.31\linewidth}
	\includegraphics[scale=0.45]{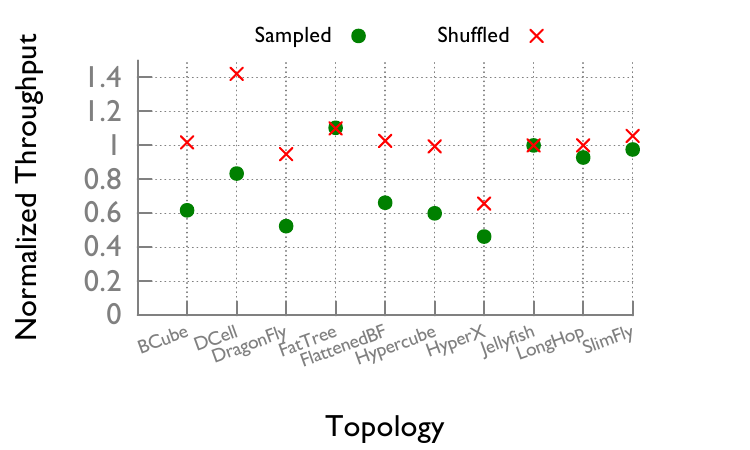}
			\vspace{-2mm}		
			\centering
			\caption{\small Comparison of topologies with TM-F(~\cite{FBTM})}
			\label{fig:AllCompareFB2}
	\end{minipage}
	    \begin{minipage}[t]{0.02\linewidth}~
  \end{minipage}
	  \begin{minipage}[b]{0.31\linewidth}
	\includegraphics[width=2in, height=1.3in]{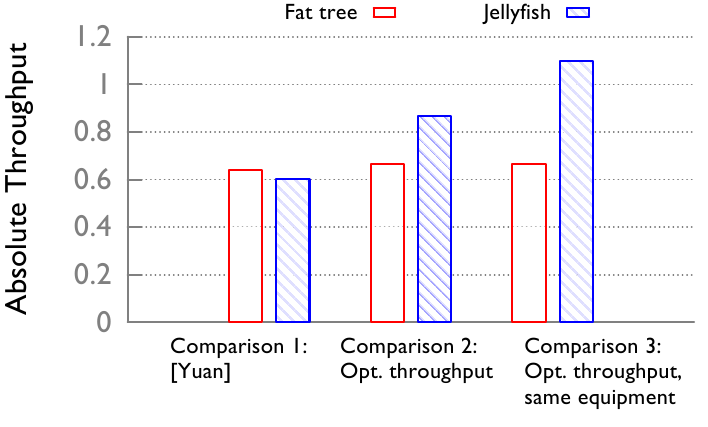}
			\vspace{-2mm}
			\centering
			\caption{\small \SC{Throughput comparison of Fat Tree and Jellyfish based on~\cite{HPC13}, showing the effect of two methodological changes}}
			\label{fig:hpc}
	\end{minipage}
\end{figure*}

\subsection{Summary} 
We have experimentally shown that:
\begin{itemize}[leftmargin=*]
	\item \pbgnewer{Cuts (as estimated using the best of a variety of heuristics) differ significantly from worst-case throughput in the majority of networks.}  
	\item Our longest matching TM approaches the theoretical lower bound in all topology families. They are significantly more difficult to route than all-to-all and random matching TMs.
	\item At large sizes, Jellyfish, Long Hop and Slim Fly networks provide the best throughput performance. 
	\item Designing high throughput networks based on proxies for throughput may not yield the expected results (demonstrated with bisection bandwidth in HyperX). 
	\item All networks except fat trees have graceful degradation in throughput under non-uniform TMs when the fraction of large flows increases. Poor performance of fat trees is due to heavy load at the ToRs.
	\item Randomizing the placement of heavy flows could improve the performance for throughput-intensive applications.

\end{itemize}

%



\section{Comparison with Related Work}
\label{sec:related}

The literature on network topology design is large and growing, with a number of designs having been proposed in the past few years~\cite{fattree-new, vl2, BCube, portland, dcell, cthrough, helios, proteus, jellyfish, legup, rewire}. However, each of these research proposals only makes a comparison with one or two other past proposals, with no standard benchmarks for the comparison. There has been no rigorous evaluation of a large number of topologies. 

Popa et.~al.~\cite{lucian10} assess $4$ topologies to determine the one that incurs least expense while achieving a target level of performance under a specific workload (all-to-all traffic). Their attempts to equalize performance required careful calibration, and approximations still had to be made. Accounting for the different costs of building different topologies is an inexact process. We sidestep that issue by using the random graph as a normalizer: instead of attempting to match performance, for each topology, we build a random graph with \SC{\emph{identical equipment}}, and then compare throughput performance of the topology with that of the random graph. This alleviates the need for equalizing configurations, thereby making it easy for others to use our tools, and to test arbitrary workloads.


\pbgnewer{Yuan et al.~\cite{HPC13} compared the performance of fat trees and Jellyfish with HPC traffic and a specific routing scheme, concluding they have similar performance, which disagrees with our results.  We replicated the method of~\cite{HPC13} and were able to reproduce their result of similar throughput using the A2A TM (Fig.~\ref{fig:hpc}, Comparison 1).  Their method splits flows into sub-flows and routes them along certain paths according to a routing scheme called LLSKR.  It then estimates each subflow's throughput by counting and inverting the maximum number of intersecting subflows at a link along the path.  Our LP-based exact throughput computation with the same LLSKR path restrictions yields slightly higher throughput in fat trees and substantially higher in Jellyfish, now $30\%$ greater than fat trees (Fig.~\ref{fig:hpc}, Comparison 2).  This is even though we now maximize the \emph{minimum} flow, while~\cite{HPC13} measured the average.  A second issue with the comparison in~\cite{HPC13} is that more servers are added to Jellyfish ($160$ servers) than to fat tree ($128$ servers) with $80$ switches in each topology. Equalizing all equipment, with $80$ switches and $128$ servers in both topologies, increases the performance gap to $65\%$ (Comparison 3).}

\pbgnewer{In a later work, Yuan et. al.~\cite{LFTI} compared performance of multiple topologies.  The throughput calculation was similar to~\cite{HPC13} discussed above so it may differ from an exact LP-based method.  In addition,~\cite{LFTI} uses single-path routing.  As pointed out in~\cite{LFTI}, single-path routing can perform significantly differently than multipath.  Moreover, in modern enterprise private cloud data centers as well as hyperscale public clouds, multipath routing (most commonly ECMP) is standard best practice~\cite{bgpInDC, jupiter, vl2} in order to achieve acceptable performance.  In contrast to~\cite{LFTI}, our study measures performance with the best possible (multipath) flow routing, so that our results reveal the topology's limits rather than a routing scheme's limits.  Our LP-based method could easily be extended to incorporate path restrictions if these were desired (as we did in the previous paragraph).  In addition, we study worst-case and real-world measured TMs.}

Other work on comparing topologies is more focused on reliability and cuts in the topology~\cite{sandiaPaper}. For instance, REWIRE~\cite{rewire} optimizes its topology designs for high sparsest cut, although it refers to the standard sparsest cut metric as bisection bandwidth. Tomic~\cite{ANCS} builds topologies with the objective of maximizing bisection bandwidth (in a particular class of graphs). Webb et. al~\cite{webb2011topology} use bisection bandwidth to pick virtual topologies over the physical topology. An interesting point of note is that they consider all-to-all traffic ``a worst-case communication scenario'', while our results (Figure~\ref{fig:compareTMsAll}) show that other traffic patterns can be significantly worse. PAST~\cite{stephens2012past} tests $3$ data center network proposals with the same sparsest cut (while referring to it as bisection bandwidth). 
PAST finds that the throughput performance of topologies with the same sparsest cut is different in packet-level simulations, raising questions about the usefulness of such a comparison; one must either build topologies of the same cost and compare them on throughput (as we do), or build topologies with the same performance and compare cost (as Popa et. al~\cite{lucian10} do). 
These observations underscore the community's lack of clarity on the relationship between cut-metrics and throughput. A significant component of our work tackles this subject.

\section{Conclusion}
\label{sec:conclusion}


Although high throughput is a core goal of network design, measuring throughput is a subtle problem. We have put forward an improved benchmark for throughput, including a near-worst-case traffic matrix, and have used this metric to compare a large number of topologies.  When designing networks for environments where traffic patterns may be unknown or variable and performance requirements are stringent, we believe evaluation under the near-worst-case longest matching TM will be useful. We also evaluate real-world workloads and show that randomizing the traffic matrix can lead to significant improvement in performance of non-uniform workloads.

Our findings also raise interesting directions for future work.  First, our longest-matching traffic heuristically produces near-worst case performance, but does not always match the lower bound.  Is there an efficient method to produce even-worse-case traffic for any given topology, or provably approximate it within less than $2\times$? Second, can we leverage the result that rack-level randomization of workload placement can improve performance to provide better task placement in data centers, taking into account other practical restrictions on task placement?




{
\bibliographystyle{abbrv}
{
\bibliography{paper}

\begin{thebibliography}{10}

\bibitem{tool}
Topology evaluation tool.
\newblock \url{www.github.com/ankitsingla/topobench}.

\bibitem{hyperx}
J.~H. Ahn, N.~Binkert, A.~Davis, M.~McLaren, and R.~S. Schreiber.
\newblock Hyperx: Topology, routing, and packaging of efficient large-scale
  networks.
\newblock In {\em Proceedings of the Conference on High Performance Computing
  Networking, Storage and Analysis}, SC '09, 2009.

\bibitem{fattree-new}
M.~Al-Fares, A.~Loukissas, and A.~Vahdat.
\newblock A scalable, commodity data center network architecture.
\newblock In {\em SIGCOMM}, 2008.

\bibitem{vlbopt}
M.~Babaioff and J.~Chuang.
\newblock On the optimality and interconnection of valiant load-balancing
  networks.
\newblock In {\em INFOCOM}, 2007.

\bibitem{slim-fly}
M.~Besta and T.~Hoefler.
\newblock {Slim Fly: A Cost Effective Low-Diameter Network Topology}.
\newblock In {\em IEEE/ACM International Conference on High Performance
  Computing, Networking, Storage and Analysis (SC14)}, Nov. 2014.

\bibitem{hypercubenetwork}
L.~N. Bhuyan and D.~P. Agrawal.
\newblock Generalized hypercube and hyperbus structures for a computer network.
\newblock {\em IEEE Tran. on Computers}, 1984.

\bibitem{chawla2006hardness}
S.~Chawla, R.~Krauthgamer, R.~Kumar, Y.~Rabani, and D.~Sivakumar.
\newblock On the hardness of approximating multicut and sparsest-cut.
\newblock {\em computational complexity}, 2006.

\bibitem{hardness}
C.~Chekuri.
\newblock Routing and network design with robustness to changing or uncertain
  traffic demands.
\newblock {\em SIGACT News}, 38(3):106--129, Sept. 2007.

\bibitem{cheegerIneq}
F.~Chung.
\newblock Laplacians of graphs and cheeger's inequalities.
\newblock In {\em Combinatorics, Paul Erdös is Eighty, Vol. 2}, pages
  157--172. Janos Bolyai Mathematical Society, Budapest, 1996.

\bibitem{rewire}
A.~Curtis, T.~Carpenter, M.~Elsheikh, A.~Lopez-Ortiz, and S.~Keshav.
\newblock {REWIRE: An optimization-based framework for unstructured data center
  network design}.
\newblock In {\em INFOCOM, 2012 Proceedings IEEE}, pages 1116--1124, March
  2012.

\bibitem{legup}
A.~R. Curtis, S.~Keshav, and A.~Lopez-Ortiz.
\newblock {LEGUP}: using heterogeneity to reduce the cost of data center
  network upgrades.
\newblock In {\em CoNEXT}, 2010.

\bibitem{max-min2}
P.~Elias, A.~Feinstein, and C.~Shannon.
\newblock A note on the maximum flow through a network.
\newblock {\em Information Theory, IEEE Transactions on}, 2(4):117--119, Dec
  1956.

\bibitem{helios}
N.~Farrington, G.~Porter, S.~Radhakrishnan, H.~H. Bazzaz, V.~Subramanya,
  Y.~Fainman, G.~Papen, and A.~Vahdat.
\newblock Helios: A hybrid electrical/optical switch architecture for modular
  data centers.
\newblock In {\em SIGCOMM}, 2010.

\bibitem{max-min1}
L.~R. Ford and D.~R. Fulkerson.
\newblock {Maximal flow through a network}.
\newblock {\em Canadian Journal of Mathematics}, 8:399 --404, 1956.

\bibitem{bisectionNPComplete}
M.~R. Garey and D.~S. Johnson.
\newblock {\em Computers and Intractability: A Guide to the Theory of
  NP-Completeness}.
\newblock 1979.

\bibitem{vl2}
A.~Greenberg, J.~R. Hamilton, N.~Jain, S.~Kandula, C.~Kim, P.~Lahiri, D.~A.
  Maltz, P.~Patel, and S.~Sengupta.
\newblock {VL2: A Scalable and Flexible Data Center Network}.
\newblock In {\em SIGCOMM}, 2009.

\bibitem{BCube}
C.~Guo, G.~Lu, D.~Li, H.~Wu, X.~Zhang, Y.~Shi, C.~Tian, Y.~Zhang, and S.~Lu.
\newblock B{C}ube: a high performance, server-centric network architecture for
  modular data centers.
\newblock {\em SIGCOMM Comput. Commun. Rev.}, 39(4):63--74, Aug. 2009.

\bibitem{dcell}
C.~Guo, H.~Wu, K.~Tan, L.~Shi, Y.~Zhang, and S.~Lu.
\newblock {DCell: A scalable and fault-tolerant network structure for data
  centers}.
\newblock In {\em SIGCOMM}, 2008.

\bibitem{gurobi}
{Gurobi Optimization Inc.}
\newblock Gurobi optimizer reference manual.
\newblock \url{http://www.gurobi.com}, 2013.

\bibitem{netNoise}
T.~Hoefler, T.~Schneider, and A.~Lumsdaine.
\newblock The impact of network noise at large-scale communication performance.
\newblock In {\em Parallel Distributed Processing, 2009. IPDPS 2009. IEEE
  International Symposium on}, pages 1--8, May 2009.

\bibitem{hpcContention}
A.~Jokanovic, G.~Rodriguez, J.~C. Sancho, and J.~Labarta.
\newblock Impact of inter-application contention in current and future hpc
  systems.
\newblock In {\em High Performance Interconnects (HOTI), 2010 IEEE 18th Annual
  Symposium on}, pages 15--24, Aug 2010.

\bibitem{hpcQos}
A.~Jokanovic, J.~C. Sancho, J.~Labarta, G.~Rodriguez, and C.~Minkenberg.
\newblock Effective quality-of-service policy for capacity high-performance
  computing systems.
\newblock In {\em High Performance Computing and Communication 2012 IEEE 9th
  International Conference on Embedded Software and Systems (HPCC-ICESS), 2012
  IEEE 14th International Conference on}, pages 598--607, June 2012.

\bibitem{hpcInterferenceAvoidance}
A.~Jokanovic, J.~C. Sancho, G.~Rodriguez, A.~Lucero, C.~Minkenberg, and
  J.~Labarta.
\newblock Quiet neighborhoods: Key to protect job performance predictability.
\newblock In {\em Parallel and Distributed Processing Symposium (IPDPS), 2015
  IEEE International}, pages 449--459, May 2015.

\bibitem{flatBF}
J.~Kim, W.~J. Dally, and D.~Abts.
\newblock Flattened butterfly: a cost-efficient topology for high-radix
  networks.
\newblock {\em SIGARCH Comput. Archit. News}, 35(2):126--137, June 2007.

\bibitem{dragonfly}
J.~Kim, W.~J. Dally, S.~Scott, and D.~Abts.
\newblock {Technology-Driven, Highly-Scalable Dragonfly Topology}.
\newblock In {\em Proceedings of the 35th Annual International Symposium on
  Computer Architecture}, ISCA '08, pages 77--88, 2008.

\bibitem{kodialam}
M.~Kodialam, T.~V. Lakshman, and S.~Sengupta.
\newblock Traffic-oblivious routing in the hose model.
\newblock {\em IEEE/ACM Trans. Netw.}, 19(3):774--787, 2011.

\bibitem{bgpInDC}
P.~Lapukhov, A.~Premji, and J.~Mitchell.
\newblock Use of bgp for routing in large-scale data centers.
\newblock Internet-Draft draft-ietf-rtgwg-bgp-routing-large-dc-09, IETF
  Secretariat, March 2016.
\newblock
  \url{http://www.ietf.org/internet-drafts/draft-ietf-rtgwg-bgp-routing-large-dc-09.txt}.

\bibitem{sandiaPaper}
D.~S. Lee and J.~L. Kalb.
\newblock Network topology analysis.
\newblock Technical report, 2008.

\bibitem{max-min-bound}
T.~Leighton and S.~Rao.
\newblock Multicommodity max-flow min-cut theorems and their use in designing
  approximation algorithms.
\newblock {\em J. ACM}, 46(6):787--832, Nov. 1999.

\bibitem{fatTree}
C.~E. Leiserson.
\newblock Fat-trees: universal networks for hardware-efficient supercomputing.
\newblock {\em IEEE Trans. Comput.}, 34(10):892--901, Oct. 1985.

\bibitem{portland}
R.~N. Mysore, A.~Pamboris, N.~Farrington, N.~Huang, P.~Miri, S.~Radhakrishnan,
  V.~Subramanya, and A.~Vahdat.
\newblock Portland: A scalable fault-tolerant layer 2 data center network
  fabric.
\newblock In {\em SIGCOMM}, 2009.

\bibitem{padua2011encyclopedia}
D.~Padua.
\newblock {\em Encyclopedia of Parallel Computing}.
\newblock Number v. 4 in Springer reference. Springer, 2011.
\newblock See bisection bandwidth discussion on p. 974.

\bibitem{lucian10}
L.~Popa, S.~Ratnasamy, G.~Iannaccone, A.~Krishnamurthy, and I.~Stoica.
\newblock A cost comparison of datacenter network architectures.
\newblock In {\em CoNEXT}, 2010.

\bibitem{dragonflyStencil}
B.~Prisacari, G.~Rodriguez, P.~Heidelberger, D.~Chen, C.~Minkenberg, and
  T.~Hoefler.
\newblock {Efficient Task Placement and Routing in Dragonfly Networks }.
\newblock In {\em Proceedings of the 23rd ACM International Symposium on
  High-Performance Parallel and Distributed Computing (HPDC'14)}. ACM, Jun.
  2014.

\bibitem{FBTM}
A.~Roy, H.~Zeng, J.~Bagga, G.~Porter, and A.~C. Snoeren.
\newblock Inside the social network's (datacenter) network.
\newblock {\em SIGCOMM Comput. Commun. Rev.}, 45(5):123--137, 2015.

\bibitem{maxconcflow}
F.~Shahrokhi and D.~Matula.
\newblock The maximum concurrent flow problem.
\newblock {\em Journal of the ACM}, 37(2):318--334, 1990.

\bibitem{jupiter}
A.~Singh, J.~Ong, A.~Agarwal, G.~Anderson, A.~Armistead, R.~Bannon, S.~Boving,
  G.~Desai, B.~Felderman, P.~Germano, A.~Kanagala, J.~Provost, J.~Simmons,
  E.~Tanda, J.~Wanderer, U.~H\"{o}lzle, S.~Stuart, and A.~Vahdat.
\newblock Jupiter rising: A decade of clos topologies and centralized control
  in google's datacenter network.
\newblock In {\em Proceedings of the 2015 ACM Conference on Special Interest
  Group on Data Communication}, SIGCOMM '15, pages 183--197. ACM, 2015.

\bibitem{singla14throughput}
A.~Singla, P.~B. Godfrey, and A.~Kolla.
\newblock High throughput data center topology design.
\newblock In {\em 11th USENIX Symposium on Networked Systems Design and
  Implementation (NSDI)}, April 2014.

\bibitem{jellyfish}
A.~Singla, C.-Y. Hong, L.~Popa, and P.~B. Godfrey.
\newblock Jellyfish: Networking data centers randomly.
\newblock In {\em NSDI}, 2012.

\bibitem{proteus}
A.~Singla, A.~Singh, K.~Ramachandran, L.~Xu, and Y.~Zhang.
\newblock Proteus: a topology malleable data center network.
\newblock In {\em HotNets}, 2010.

\bibitem{stephens2012past}
B.~Stephens, A.~Cox, W.~Felter, C.~Dixon, and J.~Carter.
\newblock {PAST: Scalable Ethernet for data centers}.
\newblock In {\em {CoNEXT}}, 2012.

\bibitem{ANCS}
R.~V. Tomic.
\newblock {Optimal networks from error correcting codes}.
\newblock In {\em ACM/IEEE Symposium on Architectures for Networking and
  Communications Systems (ANCS)}, 2013.

\bibitem{tornado}
B.~Towles and W.~J. Dally.
\newblock Worst-case traffic for oblivious routing functions.
\newblock In {\em Proceedings of the Fourteenth Annual ACM Symposium on
  Parallel Algorithms and Architectures}, SPAA '02, pages 1--8, New York, NY,
  USA, 2002. ACM.

\bibitem{Xpander}
A.~Valadarsky, M.~Dinitz, and M.~Schapira.
\newblock Xpander: Unveiling the secrets of high-performance datacenters.
\newblock In {\em Proceedings of the 14th ACM Workshop on Hot Topics in
  Networks}, HotNets-XIV, pages 16:1--16:7, New York, NY, USA, 2015. ACM.

\bibitem{cthrough}
G.~Wang, D.~G. Andersen, M.~Kaminsky, K.~Papagiannaki, T.~S.~E. Ng, M.~Kozuch,
  and M.~Ryan.
\newblock {c-Through: Part-time Optics in Data Centers}.
\newblock In {\em SIGCOMM}, 2010.

\bibitem{webb2011topology}
K.~C. Webb, A.~C. Snoeren, and K.~Yocum.
\newblock Topology switching for data center networks.
\newblock In {\em HotICE}, 2011.

\bibitem{LFTI}
X.~Yuan, S.~Mahapatra, M.~Lang, and S.~Pakin.
\newblock {LFTI:} {A} new performance metric for assessing interconnect designs
  for extreme-scale {HPC} systems.
\newblock In {\em 2014 {IEEE} 28th International Parallel and Distributed
  Processing Symposium, Phoenix, AZ, USA, May 19-23, 2014}, pages 273--282,
  2014.

\bibitem{HPC13}
X.~Yuan, S.~Mahapatra, W.~Nienaber, S.~Pakin, and M.~Lang.
\newblock {A New Routing Scheme for Jellyfish and Its Performance with HPC
  Workloads}.
\newblock In {\em Proceedings of the International Conference on High
  Performance Computing, Networking, Storage and Analysis}, SC '13, pages
  36:1--36:11, 2013.

\end{thebibliography}
}
}

{\normalfont\bfseries\MakeUppercase { Appendix}}

\label{appendix}

{\normalfont\bfseries\MakeUppercase { A. Proof of Theorem 1}}

We revisit the \emph{maximum concurrent flow} problem, based on which we defined throughput in \S\ref{sec:Metrics}: Given a network $G=(V,E_G)$ with capacities $c(u,v)$ for every edge $(u,v)\in E_G$, and a collection (not necessarily disjoint) of pairs $(s_i,t_i), i=1,\ldots,k$ each having a unit flow demand, we are interested in maximizing the minimum flow. Instead of the traffic matrix (TM) formulation of \S\ref{sec:Metrics}, for the following discussion, it will be convenient to think of the pairs of vertices that require flow between them as defining a demand graph, $H=(V,E_H)$. Thus, given $G$ and $H$, we want the maximum throughput. As we noted in \S\ref{sec:Metrics}, this problem can be formulated as a standard linear program, and is thus computable in polynomial time.

We are interested in comparing our suggested throughput metric with sparsest cut. We first prove the following theorem.

\begin{theorem}
The dual of the linear program for computing throughput is a linear programming relaxation for sparsest cut.
\end{theorem}
\begin{proof}
We shall use a formulation of the throughput linear program that involves an exponential number of variables but for which is easier to derive the dual. We denote by $P_{s,t}$ the set of all paths from $s$ to $t$ in $G$ and we introduce a variable $x_p$ for each path $p\in P_{s,t}$, for each $(s,t)\in E_H$, corresponding to how many units of flow from $s$ to $t$ are routed through path $p$.

\begin{align*}
&\text{max} &&y&& \\
&\text{subject to}&& \sum_{p\in P_{s,t}} x_p\geq y && \forall (s,t)\in E_H, \\
&&&\sum_{p:(u,v)\in p}x_p\leq c(u,v) && \forall (u,v)\in E_G \\
&&&x_p\geq 0 && \forall p\\
&&&y\geq 0.
\end{align*}

The dual of the above linear program will have one variable $w(s,t)$ for each $(s,t)\in E_H$ and one variable $z(u,v)$ for each $(u,v)\in E_G$.

\begin{align*}
&\text{min} &&\sum_{u,v}z(u,v)c(u,v)&& \\
&\text{subject to}&& \sum_{(s,t)\in E_H} w(s,t)\geq 1 \\
&&&\sum_{(u,v)\in p}z(u,v)\geq w(s,t) && \forall (s,t)\in E_H, p\in P_{s,t} \\
&&&w(s,t)\geq 0 && \forall (s,t)\in E_H\\
&&&z(u,v)\geq 0 && \forall (u,v) \in E_G.
\end{align*}

It is not hard to realize that in an optimal solution, without loss of generality, $w(s,t)$ is the length of the shortest path from $s$ to $t$ in the graph weighted by the $z(u,v)$. We can also observe that in an optimal solution we have $\sum w(s,t) =1$. These remarks imply that the above dual is equivalent to the following program, where we introduce a variable $l(x,y)$ for every pair or vertices in $E_G\cup E_H$.

\begin{align*}
&\text{min} &&\sum_{u,v}l(u,v)c(u,v)&& \\
&\text{subject to}&& \sum_{(s,t)\in E_H} l(s,t)= 1 \\
&&&\sum_{(u,v)\in p}l(u,v)\geq l(s,t) && \forall (s,t)\in E_H, p\in P_{s,t} \\
&&&l(u,v)\geq 0 && \forall (u,v)\in E_G\cup E_H\\
\end{align*}

The constraints $\sum_{(u,v)\in p}l(u,v)\geq l(s,t)$ can be equivalently restated as triangle inequalities. This means that we require $l(u,v)$ to be a metric over $V$. These observations give us one more alternative formulation:

\begin{equation}\label{eq:metrics}
\text{min}_{l(\cdot,\cdot) \text{ metric}}\frac{\sum_{(u,v)\in E_G}c(u,v)\cdot l(u,v)}{\sum_{(s,t)\in E_H}l(s,t)}
\end{equation}

We can finally see that the above formulation is a linear programming relaxation for a cut problem. More specifically, the sparsest cut problem is asking to find a cut $S$ that minimizes the ratio
\begin{equation}\label{eq:cuts}
\frac{\sum_{(u,v)\in E_G \text{ cut by S}}c(u,v)}{|\text{ edges} \in E_H \text{ cut by S}|}
\end{equation}

This is equivalent to minimizing ratio (\ref{eq:metrics}) over $\ell_1$ metrics only.

If we take $E_H$ to be the complete graph (corresponding to all-to-all demands), we get the standard sparsest cut definition:
\begin{equation}\label{eq:spcut}
\frac{\sum_{(u,v)\in E_G \text{ cut by S}}c(u,v)}{|S||\bar{S}|}
\end{equation}
\end{proof}

Before we prove Theorem~\ref{thm:graphB} from \S\ref{sec:Metrics}, we shall demonstrate the following claim.

\begin{claim}
If $G$ is a $d$-regular expander graph on $N$ nodes and $H$ is the complete graph, the value of the linear program for throughput is $\O(\frac{d \log d}{N \log N})$. The value of the sparsest cut is $\Omega(\frac{d}{N})$.
\end{claim}
\begin{proof}
Let us denote by $T$ the optimal value of expression (\ref{eq:metrics}). Note that this is the optimal value of the dual for the linear program for throughput, therefore equal to the optimal throughput. By taking $l(\cdot,cdot)$ to be the shortest path metric on $G$, we calculate:

\begin{equation}\label{eqn:tpexpander}
T\leq \frac{\sum_{(i,j)\in E_G}l(i,j)}{\sum_{i,j \in V}l(i,j)}\leq \frac{d/2. |V|}{\Theta(N^2 \frac{\log N}{\log d})}\leq O(\frac{d\log d}{N^2 \log N})
\end{equation}

Here, the first inequality follows from the fact that for $d$-regular graphs, each node can reach no more than $d^i$ nodes in $i$ hops. This means that given a vertex $v$, there exist $\Theta(N)$ nodes with distance at least $\frac{\log N}{\log d}$ from it, which means that the total distance between all pairs of nodes is $\Theta(N^2 \frac{\log N}{\log d})$.\\

Let $\Phi$ denote the minimum value of ratio (\ref{eq:spcut}) for $G$. Since $G$ is an expander, this ratio is

\begin{equation}\label{eqn:cutexpander}
\Phi \geq \Omega(\frac{d\cdot |S|}{|S||V-S|})=\Omega(\frac{d}{N})
\end{equation}
\end{proof}

\noindent {\bf Theorem \ref{thm:graphB}.} {\em Let graph $G$ be any $2d$-regular expander on $N=\frac{n}{dp}$ nodes, where $d$ is a constant and $p$ is a free parameter. Let graph B be constructed by replacing each edge of $G$ with a path of length $p$. Then, $B$ has throughput $T_B = O(\frac{1}{np{\log n}})$ and sparsest cut $\Phi_B=\Omega(\frac{1}{np})$.}
\begin{proof}
Let $(S_1,S_2)$ be the sparsest cut in $B$. Let $({S_1}',{S_2}')$ be the corresponding cut in $G$. Namely, if an edge was cut in $B$ by $(S_1,S_2)$ that belonged to a path $p_e$ then $({S_1}',{S_2}')$ cuts $e$. Let $\Phi_B$ be the value of the cut $(S_1,S_2)$ in $A$ and $\Phi_G$ the value of $({S_1}',{S_2}')$ in $G$. Then

$$\Phi_B=\frac{E(S_1,S_2)}{|S_1||S_2|}=\frac{E({S_1}',{S_2}')}{|S_1||S_2|}\geq \frac{E({S_1}',{S_2}')}{p\cdot |S_1'|p\cdot |S_2'|}\geq \frac{\Phi_G}{p^2}$$
by equation (\ref{eqn:cutexpander}) we have $\Phi_G \geq \Omega(\frac{1}{N})=\Omega(\frac{p}{n})$ which gives us $$\Phi_B\geq \Omega(\frac{1}{np})$$

On the other hand, let $T_B$ be the value of the throughput of $B$. We follow a similar reasoning as we did in equation (\ref{eqn:tpexpander}).
\begin{equation}
\begin{split}
T_B\leq \frac{\sum_{(i,j)\in E_G}l(i,j)}{\sum_{i,j \in V}l(i,j)}\leq \frac{Ndp}{\Theta((Np)^2 p\log N)} \\
\leq O(\frac{1}{Np^2\log N})=O(\frac{1}{np\log n})
\end{split}
\end{equation}
\end{proof}

\label{sec:thm2proof}

{\normalfont\bfseries\MakeUppercase { B. Proof of Theorem 2}}

\begin{proof}
$T_{A2A}$ has demand $\frac{1}{n}$ on each flow, so the largest feasible multicommodity flow routing of $T_{A2A}$ in $G$ has capacity $\frac{t}{n}$ allocated to each flow. Let $C$ be a graph representing this routing, i.e., a complete digraph with capacity $\frac{t}{n}$ on each link.  Systems-oriented readers may find it useful to think of $C$ as an overlay network implemented with reserved bandwidth in $G$.  In other words, to prove the theorem, it is sufficient to show that taking $T$ to be any hose-model traffic matrix, $T\cdot t/2$ is feasible in $C$.

For this, we use a two-hop routing scheme analogous to Valiant load balancing~\cite{vlbopt}. Consider any traffic demand $v\leadsto w$. In the first step, we split this demand flow into $n$ equal parts, routing flow $\frac{1}{n}\cdot T(v,w)\cdot t/2$ from $v$ to every node in the network, along the direct links (or the zero-hop path when the target is $v$ itself).  In the second step, the traffic arriving at each node is sent along at most one link to its final destination.

We now have to show that this routing is feasible in $C$.  Consider any link $i\to j$. This link will carry a fraction $\frac{1}{n}$ of all the traffic originated by $i$, and a fraction $\frac{1}{n}$ of all the traffic destined to $j$.  Because $T$ is a hose model TM, each node originates and sinks a total of $\leq 1$ unit of traffic; and since we are actually attempting to route the scaled traffic matrix $T\cdot t/2$, each node originates and sinks a total of $\leq t/2$ units of traffic.  Therefore, link $i$ carries a total of \[\frac{t}{2} \cdot \frac{1}{n} + \frac{t}{2} \cdot \frac{1}{n} = \frac{t}{n},\]
which is the available capacity on each link of $C$ and is hence feasible.
\end{proof}

\label{sec:cutEval}

{\normalfont\bfseries\MakeUppercase { C. Measuring cuts}}

We employ several heuristics for estimating sparsest cut. 

\paragraphb{Brute-force computation} Brute force computation of sparsest cut is computationally intensive since it considers all possible cuts in the network ($2^{n-1}$ cuts in a network with $n$ nodes). In addition to bandwidth, the number of flows traversing each cut has to be estimated which adds further overhead in the computation of sparsest cut.

Due to the computational complexity, brute force evaluation of sparsest cut is possible only for networks of size less than $20$. However, we perform limited brute-force computation on all networks by capping the computation at $10,000$ cuts.

\paragraphb{One-node cuts} Designed computer networks as well as naturally occurring networks tend to be denser at the core and sparse at the edges. When the core has high capacity, it is likely that the worst-case cut occurs at the edges. Hence, this heuristic considers all cuts with only a single node in a subset formed by the cut. There exists $n$ cuts with a single node. This is a very small fraction of the total $2^{n-1}$ cuts.

\paragraphb{Two-node cuts} $\frac{n * (n-1)}{2}$ cuts with two nodes in a subset also reveal the limited connectivity at the edges of the network.

 \begin{table*}[h]

\small
\centering
  \begin{tabular}{ | l | c | c | c | c | c | c | c | c |}
    \hline
    & &  & \multicolumn{5}{c|}{Sparsest cut estimator which found the worst cut} \\ 
    Topology family & Total & \parbox{2cm}{\#networks with throughput= estimated cut} & Brute force & 1-node & 2-node & \parbox{1.5cm}{Expanding regions} & Eigenvector \\ \hline
    BCube & 7 & 2 & 2 & 0 & 0 & 3 & 7  \\ \hline
    DCell & 4 & 2 & 2 & 0 & 0 & 2 & 3 \\ \hline
	Dragonfly & 4 & 0 & 2 & 0 & 0 & 0 & 2 \\ \hline
	Fat tree & 8 & 8 & 8 & 8 & 8 & 8 & 8 \\ \hline 
	Flattened butterfly & 8 & 5 & 6 & 0 & 1 & 0 & 5 \\ \hline   	
    	Hypercube & 7 & 3 & 3 & 0 & 0 & 1 & 6 \\ \hline
    	HyperX & 11 & 1 & 1 & 0 & 0 & 1 & 10 \\ \hline
	Jellyfish & 350 & 3 & 0 & 0 & 0 & 2 & 349\\ \hline
	LongHop & 110 & 9 & 45 & 0 & 0 & 1 & 66\\ \hline
	SlimFly & 6 & 1 & 1 & 0 & 0 & 0 & 5\\ \hline \hline
	Natural networks & 66 & 48 & 18 & 21 & 11 & 34 & 38 \\ \hline \hline
	Total & 581 & 82 & 88 & 29 & 20 & 52 & 499  \\ \hline
    \end{tabular}
    \caption{Estimated sparsest cuts: Do they match throughput, and which estimators produced those cuts?}
    \label{table:cuts}

\end{table*}

\paragraphb{Expanding cuts} It is likely that the network is clustered,i.e., it contains two or more highly connected components connected by a few links. Subsets of all possible combinations of contiguous nodes in the network can be very large. We optimize our search to a subspace of this category of cuts. Starting from each node, we consider cuts which include nodes within a distance $k$ from the node. When $k=0$, the cut involves only the originating node and is equivalent to the single node case discussed before. When $k=1$, all nodes within distance $1$ from the node are considered -- the node and its neighbors. $k$ is incremented until the entire graph is covered. If $d$ is the diameter of the network, the number of cuts considered is less than or equal to $n * d$.

\paragraphb{Eigenvector based optimizations} Eigenvector corresponding to the second smallest eigenvalue of the normalized Laplacian of a graph can give a set of $n$ cuts, the worst of which is within a constant factor from the actual cut~\cite{cheegerIneq}. The nodes of the graph are sorted in the ascending order corresponding to their value in the second eigenvector~\cite{cheegerIneq}. We sweep this vector of sorted nodes to obtain the $n$ cuts.

How well did our sparse cut heuristics perform? Comparing columns 2 and 3 in Table~\ref{table:cuts}, we see that cuts accurately predicted throughput in less than $15\%$ of the tested networks only. Table~\ref{table:cuts} shows how often each estimator found the sparse cut.  More than one technique may find the sparse cut, hence the sum may not equal the total number of networks.  Brute-force computation was helpful in finding $15\%$ of the sparse cuts. Cuts involving one or two nodes and contiguous regions of the graph also found the sparse cut in a small fraction of networks (less than $10\%$ each). The majority of such networks are the natural networks, which are \pbgnew{often} denser in the core and sparser in the edges. Sparse connectivity at the edges lead to bottlenecks at the edge which are revealed by cuts involving one or two nodes. Fat tree is another interesting \pbgnew{case where every heuristic's cuts yield} the accurate flow value. Overall, the eigenvector-based approximation found the largest number of sparse cuts ($86\%$), but it is known not to be a tight approximation~\cite{cheegerIneq}, and the full collection of heuristics did improve on it in a nontrivial fraction of cases.

\end{document}